\documentclass[conference]{IEEEtran}
\IEEEoverridecommandlockouts
\usepackage{float}
\usepackage{caption} 
\usepackage{subcaption}
\usepackage{color}
\usepackage{graphicx}
\usepackage{epsfig}
\usepackage{url}
\usepackage{version}
\usepackage{enumitem}
\usepackage{mdframed}
\usepackage{amsmath}
\DeclareCaptionFont{tiny}{\tiny}
\usepackage{colortbl}
\usepackage{amssymb}
\usepackage{amsthm}

\definecolor{light-gray}{gray}{0.96}
\definecolor{LightCyan}{rgb}{0.88,1,1}

\newcommand{\defn}[1]{\textbf{\emph{#1}}}
\newtheorem{claim}{Claim}
\newtheorem{lemma}{Lemma}

\newcommand{\B}{\textsc{BEB}}
\newcommand{\LB}{\textsc{Log-Backoff}}
\newcommand{\LLB}{\textsc{LogLog-Backoff}}
\newcommand{\STB}{\textsc{Sawtooth-Backoff}}

\begin{document}
%
% paper title
% Titles are generally capitalized except for words such as a, an, and, as,
% at, but, by, for, in, nor, of, on, or, the, to and up, which are usually
% not capitalized unless they are the first or last word of the title.
% Linebreaks \\ can be used within to get better formatting as desired.
% Do not put math or special symbols in the title.
\title{Is Our Model for Contention Resolution Wrong?\\ \vspace{-5pt}{\LARGE Confronting the Cost of Collisions}}

\author{\IEEEauthorblockN{William C. Anderton}
\IEEEauthorblockA{Department of Computer Science and Engineering\\
Mississippi State University\\
Mississippi, 39762, USA\\
Email:  \texttt{wca36@msstate.edu}}
\and
\IEEEauthorblockN{Maxwell Young\thanks{\noindent This research is supported by the National Science Foundation grant CCF 1613772 and by a research gift from C Spire.}}
\IEEEauthorblockA{Department of Computer Science and Engineering\\
Mississippi State University\\
Mississippi, 39762, USA\\
Email:  \texttt{myoung@cse.msstate.edu}}
}

%\author{William C. Anderton, Maxwell Young% <-this % stops a space
%\thanks{William C. Anderton and Maxwell Young are with the Department of Computer Science and Engineering, Mississippi State University, MS, USA, 39762. Email: \texttt{wca36@msstate.edu} and  \texttt{myoung@cse.msstate.edu}.}% <-this % stops a space
%\thanks{This research is supported by the National Science Foundation grant CCF 1613772 and by a research gift from C Spire.}% <-this % stops a space
%}

\maketitle

\begin{abstract}
Randomized binary exponential backoff (BEB) is a popular algorithm for coordinating access to a shared channel. With an operational history exceeding four decades, BEB is currently an important component of several wireless standards. 

Despite this track record, prior theoretical results indicate that under bursty traffic (1) BEB yields poor makespan  and (2) superior algorithms are possible. To date, the degree to which these findings manifest in practice has not been resolved.

To address this issue, we examine one of the strongest cases against BEB: $n$ packets that simultaneously begin contending for the wireless channel. Using Network Simulator 3, we compare against more recent algorithms that are inspired by BEB, but whose makespan guarantees are superior. Surprisingly, we discover that these newer algorithms significantly underperform.

Through further investigation, we identify as the culprit a flawed but common abstraction regarding the cost of collisions.  Our experimental results are complemented by analytical arguments that the number of collisions -- and not solely makespan -- is an important metric to optimize. We believe that these findings have implications for the design of contention-resolution algorithms.
\end{abstract}

% make the title area
\maketitle

\IEEEpeerreviewmaketitle

\section{Introduction}\label{sec:intro}

\defn{Randomized binary exponential backoff} (\defn{BEB}) plays a critical role  in coordinating access by multiple devices to a shared communication medium. Given its importance, BEB has been studied at length and is known to yield good throughput under well-behaved traffic~\cite{GoldbergMa96a,GoldbergMaPaSr00,HastadLeRo87,RaghavanUp95,Al-Ammal2000,Al-Ammal2001,Goodman:1988:SBE:44483.44488,bianchi:performance,song:stability}. 

In contrast, when traffic is ``bursty'', BEB is suspected to perform sub-optimally.  Under a \defn{single batch} of $n$ \defn{packets} that simultaneously begin contending for the channel, Bender~et al.~\cite{BenderFaHe05} prove that BEB has  $\Theta(n\log n)$ \defn{makespan} (the amount of time until all packets are successfully transmitted). More recent  algorithms have been proposed~\cite{GreenbergFlLa87,BenderFaHe05,bender:how,bender:contention,fineman:contention,fineman:contention2,bender:heterogeneous} with improved makespan regardless of the traffic type.

Together, these results beg the question: {\it How do newer algorithms compare to BEB in practice?} Here, we make progress towards an answer by restricting ourselves to bursty traffic -- in particular, the simplest instance of such traffic: a single burst (\defn{batch}) of packets. This is a prominent case where BEB is anticipated to do poorly, and it should be a straightforward (if laborious) exercise to discover which of the following situations is true: (1) A newer contention-resolution algorithm outperforms BEB, or (2) BEB outperforms newer contention-resolution algorithms.

Interestingly, neither of these outcomes is very palatable. In one form or another, BEB has operated in networks for over four decades and it remains an essential ingredient in several wireless standards.  Bursty traffic can arise in practice~\cite{Teymori2005,yu:study} and its impact has been examined~\cite{Ghani:2010,sarkar:effect,canberk:self,bhandari:performance}. If (1) holds, then BEB is potentially in need of revision and the ramifications of this are hard to overstate.  

Conversely, if (2) holds, then theoretical results are not translating into improved performance.  At best, this is a matter of asymptotics. At worst, this indicates a problem with the abstract model upon which newer results  are based. In this latter case, it is important to understand what assumptions are faulty so that the abstract model can be revised. 

%%%%%%%%%%%%%%%%%%%%%%%%%%%%%%%%%%%
%%%%%%%%%%%%%%%%%%%%%%%%%%%%%%%%%%%
%%%%%%%%%%%%%%%%%%%%%%%%%%%%%%%%%%%

\subsection{A Common Model} 

Given $n$ stations, the problem of \defn{contention resolution} addresses the amount of time until any one of the  stations transmits alone. A natural consideration is the time until a subset of $k$ stations each transmits alone; this often falls under the same label, but is also referred to as $k$-selection~(see~\cite{Anta2010}). 

We focus on the case of $k=n$. Here, much of the algorithmic work shares an abstract model.  Three common assumptions are:
\vspace{-10pt}
\begin{figure}[h]
\begin{center}
{
\fbox{\colorbox{light-gray}{
\begin{minipage}[t]{0.42\textwidth} 

\begin{itemize}[leftmargin=3.5mm]
\item{\it A0.} Time is discretized into \defn{slots}, each of which may accommodate a packet. \vspace{3pt}
\item {\it A1.} If a single packet is transmitted in a slot, the packet \defn{succeeds}, but \defn{failure} occurs if two or more packets transmit simultaneously due to a \defn{collision}. \vspace{-8pt}
\item {\it A2.} The failure of a transmission is known to the sender with negligible delay  beyond the single slot in which the failure occurred.\end{itemize}
\end{minipage}
          }
     }     
}
\end{center}
\vspace{-13pt}
 \end{figure}

Assumption A0 is near universal, but technically inaccurate for reasons discussed in Section~\ref{sec:802.11}. To summarize, slots in a contention window are used to obtain ownership of the channel. However, transmission of the full packet may occur past this contention-window slot while all other stations pause their execution. Therefore, this assumption is sufficiently close to reality that we should not expect performance to deviate greatly as a result.

Examples of assumption A1 abound (for example~\cite{Anta2010,komlos:asymptotically,Capetanakis:2006,bender:heterogeneous,bender:contention,bender:how,fineman:contention2}), although variations exist. A  compelling alternative is the signal-to-noise-plus-interference (SINR) model~\cite{avin:sinr,moscibroda:worst} which is less strict about failure in the event of simultaneous transmissions. Another model that has received attention is the affectance model~\cite{Hall2009}. Nevertheless, these all share the reasonable assumption that simultaneous transmissions may negatively impact performance.

Assumption A2 is also widely adopted (see the same examples for A1) and implicitly addresses two quantities that affect performance: the time to transmit a packet, and the time to receive any feedback on success or failure. Assigning a delay of $1$ slot to these quantities admits a model where the problem of contention resolution is treated separately from the functionality for collision detection. Such functionality is provided by a medium access control (MAC) protocol  -- of which the contention-resolution algorithm is only one component -- and is not captured by A2.

%%%%%%%%%%%%%%%%%%%%%%%%%%%%%%%%%%%%%
%%%%%%%%%%%%%%%%%%%%%%%%%%%%%%%%%%%%%
%%%%%%%%%%%%%%%%%%%%%%%%%%%%%%%%%%%%%

\subsubsection{Demonstrating a Flawed Assumption}

Our main thesis is that A2 is  flawed in the wireless setting; that is,  the cost of failure is far more significant than the abstract model acknowledges. This is not a matter of minor adjustments to the assumption, or an artifact of hidden constants in the algorithms examined. Rather, the way in which failures -- in particular, collisions -- are detected cannot be isolated from the problem of contention resolution. 

Several corollaries follow from this thesis, all indicating that accounting for such failures should be incorporated into algorithm design. For a range of wireless settings, contention-resolution algorithms that ignore this will likely not perform as advertised  when deployed within a MAC protocol (see Section~\ref{sec:interpret}). We demonstrate this for the popular IEEE 802.11g standard.\vspace{-3pt}

%%%%%%%%%%%%%%%%%%%%%%%%%%%%%%%%%%%%%
%%%%%%%%%%%%%%%%%%%%%%%%%%%%%%%%%%%%%
%%%%%%%%%%%%%%%%%%%%%%%%%%%%%%%%%%%%%

\subsection{Overview of BEB in IEEE 802.11g}\label{sec:802.11}

To understand our findings, it is helpful to summarize IEEE 802.11g and how BEB operates within it. However, outside of this section and the description of our experimental setup,  discussion of such aspects and terminology is kept to a minimum. Throughout, we will often use interchangeably the terms {\it packets} and {\it stations} depending on the context; the two uses are equivalent given that each station seeks to transmit a single packet in the single-batch case.

Exponential backoff~\cite{MetcalfeBo76} is a widely deployed algorithm for distributed multiple access. Informally, a backoff algorithm operates over a \defn{contention window (CW)} wherein each \defn{station} makes a single randomly-timed access attempt. In the event of two or more simultaneous attempts, the result is a collision and none of the stations succeed.  Backoff seeks to avoid collisions by dynamically increasing the contention-window size such that stations succeed.

IEEE 802.11 handles contention resolution via the \defn{distributed coordination function (DCF)} which employs BEB; as the name suggests, successive CWs double in size under BEB. The operation of DCF is summarized as follows. Prior to transmitting data, a station first senses the channel for a period of time known as a \defn{distributed inter-frame space (DIFS)}. If the channel is not in use over the DIFS, the station transmits its data; otherwise, it waits until the current transmission finishes and then initiates BEB.  

%\footnote{Sensing the channel is performed by the clear channel assessment (CCA) mechanism. The details of CCA and the related network allocation vector (NAV) mechanism are not critical to our findings and we omit their details.}

For a  contention window of size CW, a timer value is selected uniformly at random from  $[0, CW-1]$. So long as the channel is sensed to be idle, the timer counts down and, when it expires, the station transmits. However, if at any prior time  the channel is sensed busy, BEB is {\it paused for the duration of the current transmission}, and then resumed (not restarted) after another DIFS.

After a station transmits, it awaits an {\bf acknowledgement (ACK)} from the receiver. If the transmission was successful, then the receiver waits for a short amount of time known as a \defn{short inter-frame space (SIFS)} -- of shorter duration than a DIFS --  before sending the ACK. Upon receiving an ACK, the station learns that its transmission was successful.  Otherwise, the station waits for an ACK-timeout duration before {\it concluding that a collision occurred}.  

This series of actions is referred to as \defn{collision detection}; the cost of which lies at the heart of our argument.  If a collision is detected, then the station must attempt a retransmission via the same process with its CW doubled.

Figure~1 illustrates the operation of DCF. Note that both the transmission of data and the acknowledgement process occur ``outside'' of the backoff component of DCF. Yet,  the focus of many algorithmic results is solely on the slots of this backoff component. % that is, and we will be making our comparison between theory and experiment over these slots.

Finally,  RTS/CTS (request-to-send and clear-to-send) is an optional mechanism. Informally, a station will send an RTS message and await an CTS message from the receiver prior to transmitting its data. Due to increased overhead, RTS/CTS is often only enabled for large packets. Therefore, we focus on the case where RTS/CTS is disabled, although our experiments show that our findings continue to hold when this mechanism is used (see Section~\ref{sec:hidden}). \vspace{-0pt}

%%%%%%%%%%%%%%%%%%%%%%%%%%%%%%%%%%%
%%%%%%%%%%%%%%%%%%%%%%%%%%%%%%%%%%%
%%%%%%%%%%%%%%%%%%%%%%%%%%%%%%%%%%%

\begin{figure}[t]
\vspace{-15pt}
\captionsetup[subfigure]{labelformat=empty}
\centering
\hspace{-0.3cm}\begin{subfigure}{0.45\textwidth} 
\includegraphics[width=1.15\textwidth]{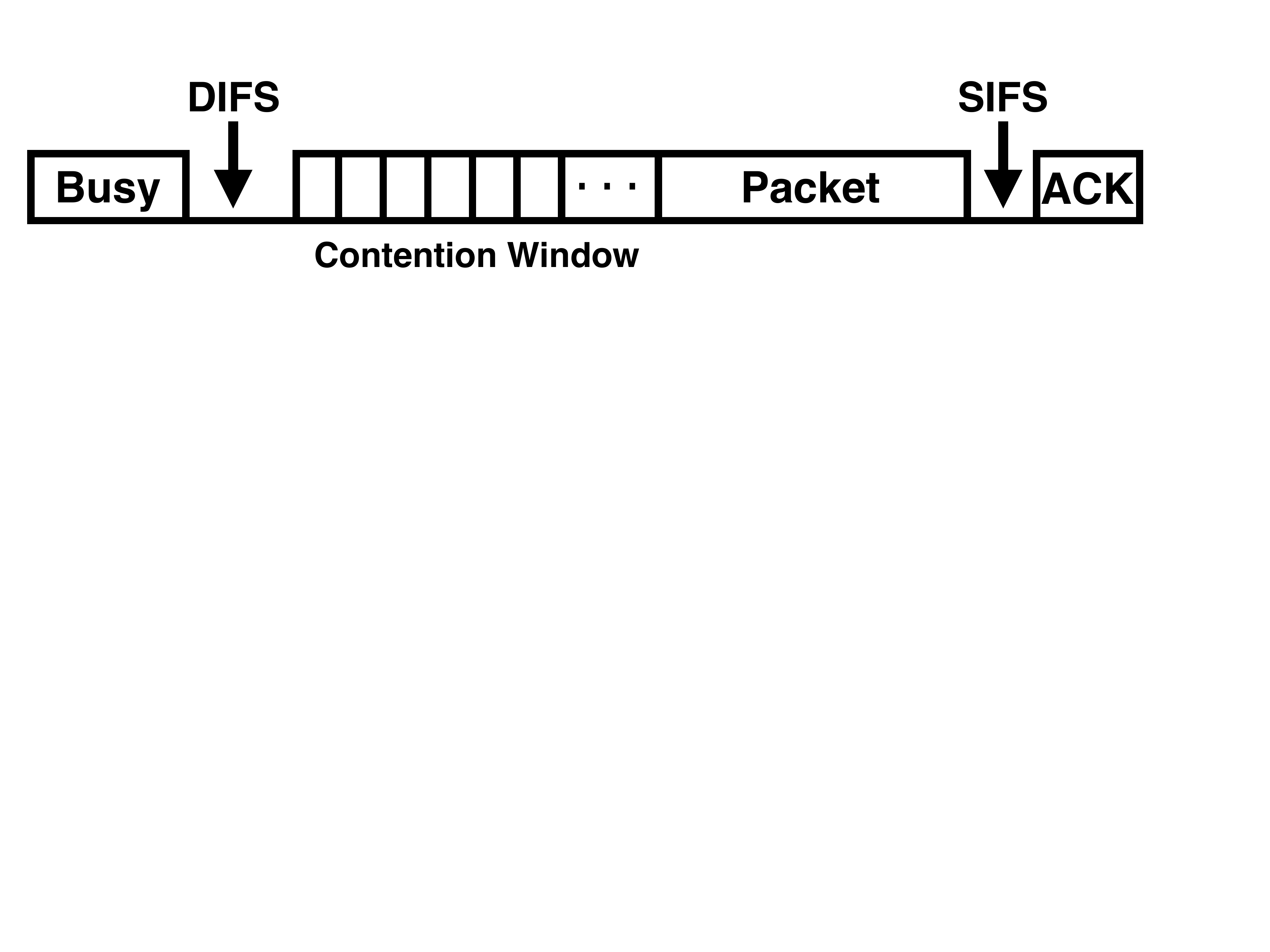} 
\vspace{-5.5cm}
\end{subfigure}
\caption*{{\bf Figure 1:} Illustration of DCF.} 
\vspace{-0.5cm}
\end{figure}

%%%%%%%%%%%%%%%%%%%%%%%%%%%%%%%%%%%
%%%%%%%%%%%%%%%%%%%%%%%%%%%%%%%%%%%
%%%%%%%%%%%%%%%%%%%%%%%%%%%%%%%%%%%

\section{Experimental Setup}\label{sec:experimental}

We employ Network Simulator 3 (NS3)~\cite{NS3} which is a widely used network simulation tool in the research community~\cite{weingartner:performance}.  Our experimental setup is described here for the purposes of reproducibility.\footnote{Our simulation code and data will be made available at \texttt{www.maxwellyoung.net}.}

% that provides an extraordinary level of detail at each layer of the protocol stack. 

%%%%%%%%%%%%%%%%%%%%%%%%%%%%%%%%%%%
%%%%%%%%%%%%%%%%%%%%%%%%%%%%%%%%%%%
%%%%%%%%%%%%%%%%%%%%%%%%%%%%%%%%%%% 
 
\begin{table} \vspace{-2pt}
\begin{center}
{
\begin{tabular}{ |p{4cm}|p{3cm}|  }
\hline
\rowcolor{LightCyan}\hspace{37pt}{\bf Parameter} &  \hspace{30pt}{\bf Value}  \\
\hline
Data rate &  54 Mbits/sec \\
\hline
Wireless specification &  802.11g \\
\hline
Slot duration & 9$\mu$s  \\
\hline
SIFS  &  16$\mu$s   \\
\hline
DIFS  &   34$\mu$s \\
\hline
ACK timeout &  75$\mu$s\\
\hline
%CRC length  & 4 bytes \\
%\hline
Preamble & 20$\mu$s  \\
\hline
Transport layer protocol & UDP\\
\hline
%Payload size & 12 or 64 bytes \\
%\hline
Packet overhead & 64 bytes\\
\hline
%PHY modulation & OFDM\\
%\hline
Contention-window size min. & 1\\
\hline
Contention-window size max.& 1024  \\
\hline
RTS/CTS & Off \\
\hline
\end{tabular}
}\vspace{-0pt}\caption{Parameter values used in our experiments.}\label{table:parameters}\vspace{-15pt}
\end{center}
\end{table}

%%%%%%%%%%%%%%%%%%%%%%%%%%%%%%%%%%%
%%%%%%%%%%%%%%%%%%%%%%%%%%%%%%%%%%%
%%%%%%%%%%%%%%%%%%%%%%%%%%%%%%%%%%%

Our reasons for using NS3 are twofold. First, wireless communication is difficult to model and employing NS3 helps allay concerns that our findings are an artifact of poorly-modeled wireless effects.  Second, given the assumptions upon which contention-resolution algorithms are based, NS3 can reveal whether we are being led astray by an assumption that appears reasonable, but results in a {\it significant} discrepancy between theory and practice.

Table~\ref{table:parameters} provides our experimental parameters. Path-loss models with default parameters are known to be faithful~\cite{stoffers:comparing} and, therefore, our experiments employ the  log-distance propagation loss model in NS3. For transmission and reception of packets (frames), we use the YANS~\cite{lacage:yans} module which provides an additive-interference model. 

At the MAC layer, we make use of IEEE 802.11g and we implement changes to the growth of the contention window based on the algorithms we investigate. All experiments use IPv4 and UDP.%\footnote{Our investigation employed UDP instead of TCP to reduce the impact of potential transport-layer effects that could have complicated the interpretation of our results. Ultimately, given the explanation for our findings, the use of TCP does not alter our conclusions.}

The amount of overhead for each packet is 64-bytes: 8 bytes for UDP, 20 bytes for IP, 8 bytes for an LLC/Snap header, and 28 bytes of additional overhead at the MAC layer. 

The duration of an acknowledgement (ACK) timeout is specified by the most recent IEEE 802.11 standard\footnote{This is a large document; please see Section 10.3.2.9, page 1317 of~\cite{802.11-standard}.} to be roughly the sum of a SIFS ($16\mu$s), standard slot time ($9\mu$s), and preamble ($20\mu$s); a total of $45\mu$s.  However, in practice, this is subject to tuning.  In our experiments, an ACK-timeout below $55\mu$s gave markedly poor performance; there is insufficient time for the ACK before the sender decides to retransmit. We use the default value of $75\mu$s in NS3 since this is the same order of magnitude and performs well. %(whose implementation also incorporates data-transmission duration)  (experiments with values between $55$ and $75$ did not alter our findings).

%of dimension. at most $40$ meters by $40$ meters.

In our experiments, $n$ stations are placed in a  $40m\times 40m$ grid, and they are laid out starting at the south-west corner of the grid moving left to right by $1$ meter increments, and then up when the current row is filled. A wireless access point (AP) is located (roughly) at the center of the grid. We do not simulate additional terrain or environmental phenomena; our goal is to test the performance under ideal conditions without complicating factors. 

Our experiments are computationally intensive.  Computing resources are provided by the High Performance Computing Collaboratory (HPC$^2$) at Mississippi State University. We employ four identical Linux (CentOS) systems, each with 16 processors (Intel Xeon CPU E5-2690, 2.90GHz) and 396 GB of memory.  %We are able to investigate system sizes up to $n=150$ using NS3. % and, as we will see, this turns out to provide useful data for our investigation. 
\vspace{-5pt}

%%%%%%%%%%%%%%%%%%%%%%%%%%%%%%%%%%%
%%%%%%%%%%%%%%%%%%%%%%%%%%%%%%%%%%%
%%%%%%%%%%%%%%%%%%%%%%%%%%%%%%%%%%%

\begin{figure}[t]\vspace{-7pt}
\begin{mdframed}
\begin{center}

\selectfont

\fbox{\hspace{-7pt}\colorbox{light-gray}{
\begin{minipage}[t]{0.99\textwidth} 

\noindent{}\hspace{-3pt}When a station wishes to transmit a packet:\vspace{-0pt}
	
          \begin{itemize}[leftmargin=3mm]%[noitemsep,nolistsep,leftmargin=9pt]\renewcommand{\labelitemii}{$\circ$}
	
          \item Set the window size $W = 1$.\vspace{-0pt}

          \item Repeat until the packet is successfully transmitted:\vspace{-0pt}
			
          	\begin{itemize}[leftmargin=3mm]%[noitemsep,nolistsep,leftmargin=7pt]\renewcommand{\labelitemii}{$\circ$}
			
                 \item Choose a slot $t$ in the window uniformly at random. Try to transmit in slot $t$.\vspace{-0pt}
			
                  \item If the transmission failed, then: (i) wait until the end of the window, and (ii) set $W \leftarrow  (1+r)W$.
			
               \end{itemize}
			
          \end{itemize}

\end{minipage}
          }
     }     
\end{center}
\end{mdframed}
\vspace{-8pt}
\caption*{{Figure 2:} Generic algorithm for LLB,~LB, and \B~where $r=1/\lg\lg W$, $r=1/\lg W$, and $r=1$, respectively.}\label{fig:generic-backoff}\vspace{-10pt}
 \end{figure}

%%%%%%%%%%%%%%%%%%%%%%%%%%%%%%%%%%%
%%%%%%%%%%%%%%%%%%%%%%%%%%%%%%%%%%%
%%%%%%%%%%%%%%%%%%%%%%%%%%%%%%%%%%%

\section{A Single Batch}\label{sec:single-batch}\vspace{-2pt}

We examine a single batch of $n$ packets that simultaneously begin their contention for the channel. As algorithmic competitors for \B, we take \LB ~(LB), \LLB~(LLB) from~\cite{BenderFaHe05} and \STB~(STB) from~\cite{Gereb-GrausT92,GreenbergL85}. Both LLB and LB are closely related to \B~in that they execute using a CW that increases in size monotonically.  The pseudocode for the algorithms LLB, LB, and BEB is provided in Figure~2. 

In contrast, STB is non-monotonic and executes over a doubly-nested loop. The outer loop sets the current window size $W$ to be double that used in the preceding outer loop; this is like BEB. Additionally, for each such $W$, the inner loop executes over $\lg W$ windows of size $W, W/2, ..., 2$ and, for each window, a slot is chosen uniformly at random for the packet to transmit; this is the ``backon'' component of STB. \vspace{-3pt}

%%%%%%%%%%%%%%%%%%%%%%%%%%%%%%%%%%%
%%%%%%%%%%%%%%%%%%%%%%%%%%%%%%%%%%%
%%%%%%%%%%%%%%%%%%%%%%%%%%%%%%%%%%%

%\rowcolor{blue!62!yellow!38}
\begin{table}[h] 
\begin{center}
{
\begin{tabular}{ |p{4cm}|c|  }
\hline
\rowcolor{LightCyan} \hspace{37pt}{\bf Algorithm} &  {\bf Contention-Window Slots}  \\
\hline
\B & $\Theta(n\log n)$ \\
\hline
\LB & $\Theta\left( \frac{n\log n}{\log\log n} \right)$  \\
\hline
\LLB & $\Theta\left( \frac{n\log\log n}{\log\log\log n} \right)$  \\
\hline
\STB & $\Theta\left(n\right)$  \\
\hline
\end{tabular}
}\vspace{-0pt}\caption{Known guarantees on CW slots for a batch of $n$ packets for BEB, LB, LLB~\cite{BenderFaHe05} and STB~\cite{Gereb-GrausT92,GreenbergL85}.}\label{table:makespan}\vspace{-7pt}
\end{center}
\end{table}

%%%%%%%%%%%%%%%%%%%%%%%%%%%%%%%%%%%
%%%%%%%%%%%%%%%%%%%%%%%%%%%%%%%%%%%
%%%%%%%%%%%%%%%%%%%%%%%%%%%%%%%%%%%

\begin{figure*}[t]\vspace{-1.8cm}
\captionsetup[subfigure]{labelformat=empty}
\centering
\hspace{-0.7cm}\begin{subfigure}{0.25\textwidth} 
\includegraphics[width=1.2\textwidth]{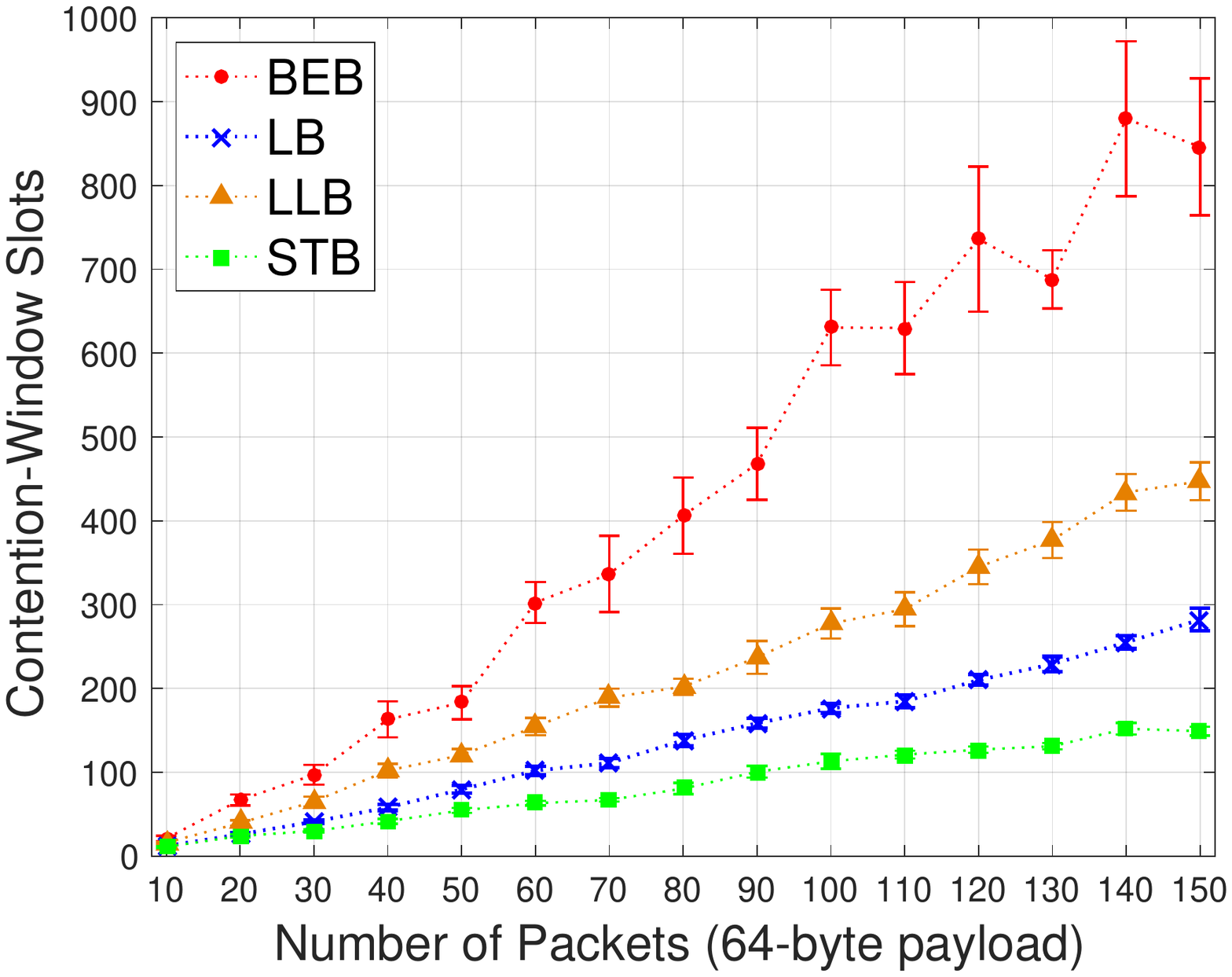} 
\vspace{-2.2cm}\caption{\hspace{29pt}(3)}
\end{subfigure}
\begin{subfigure}{0.25\textwidth} 
\includegraphics[width=1.2\textwidth]{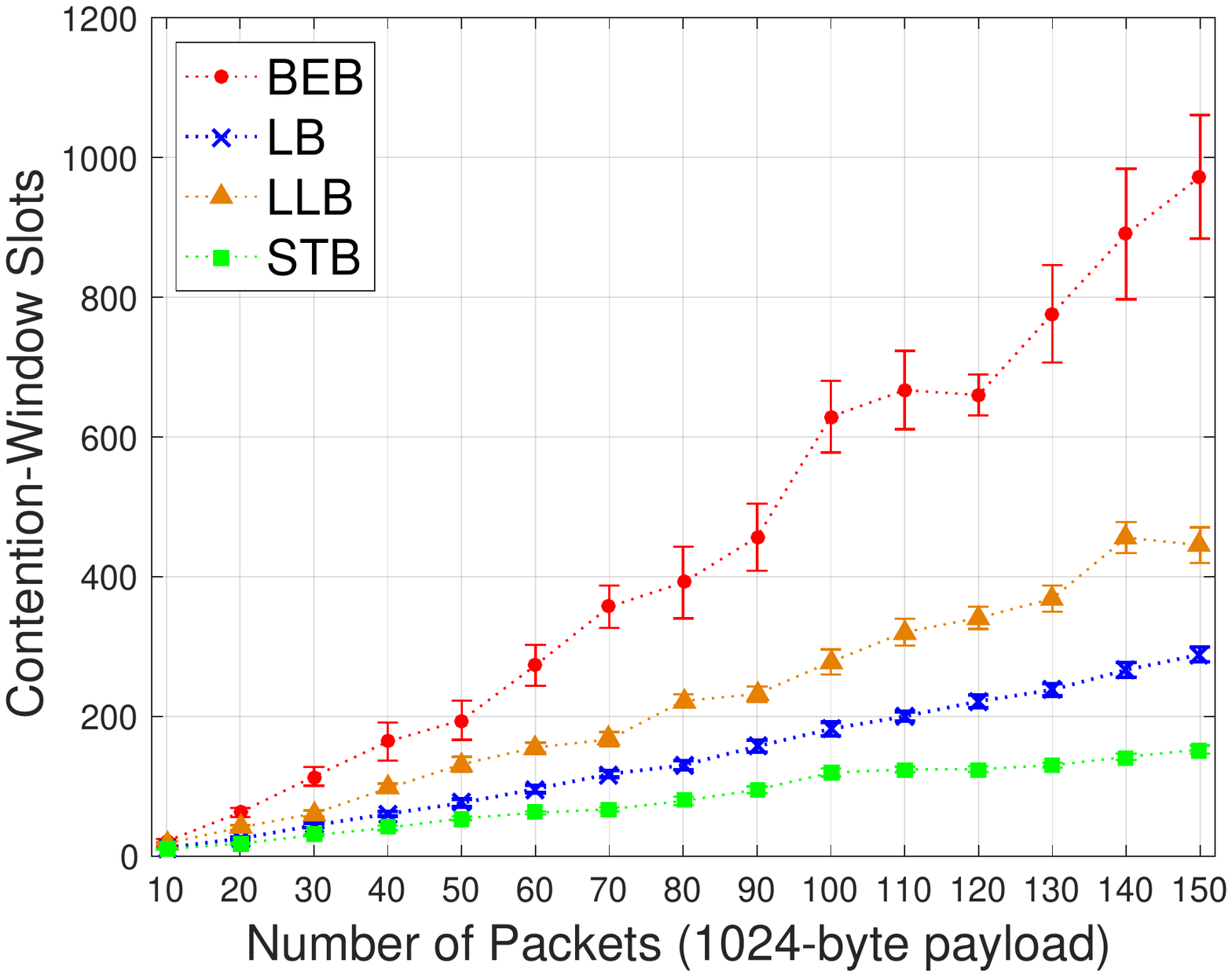} 
\vspace{-2.2cm}\caption{\hspace{29pt}(4)}
\end{subfigure}
\begin{subfigure}{0.25\textwidth}
\includegraphics[width=1.2\textwidth]{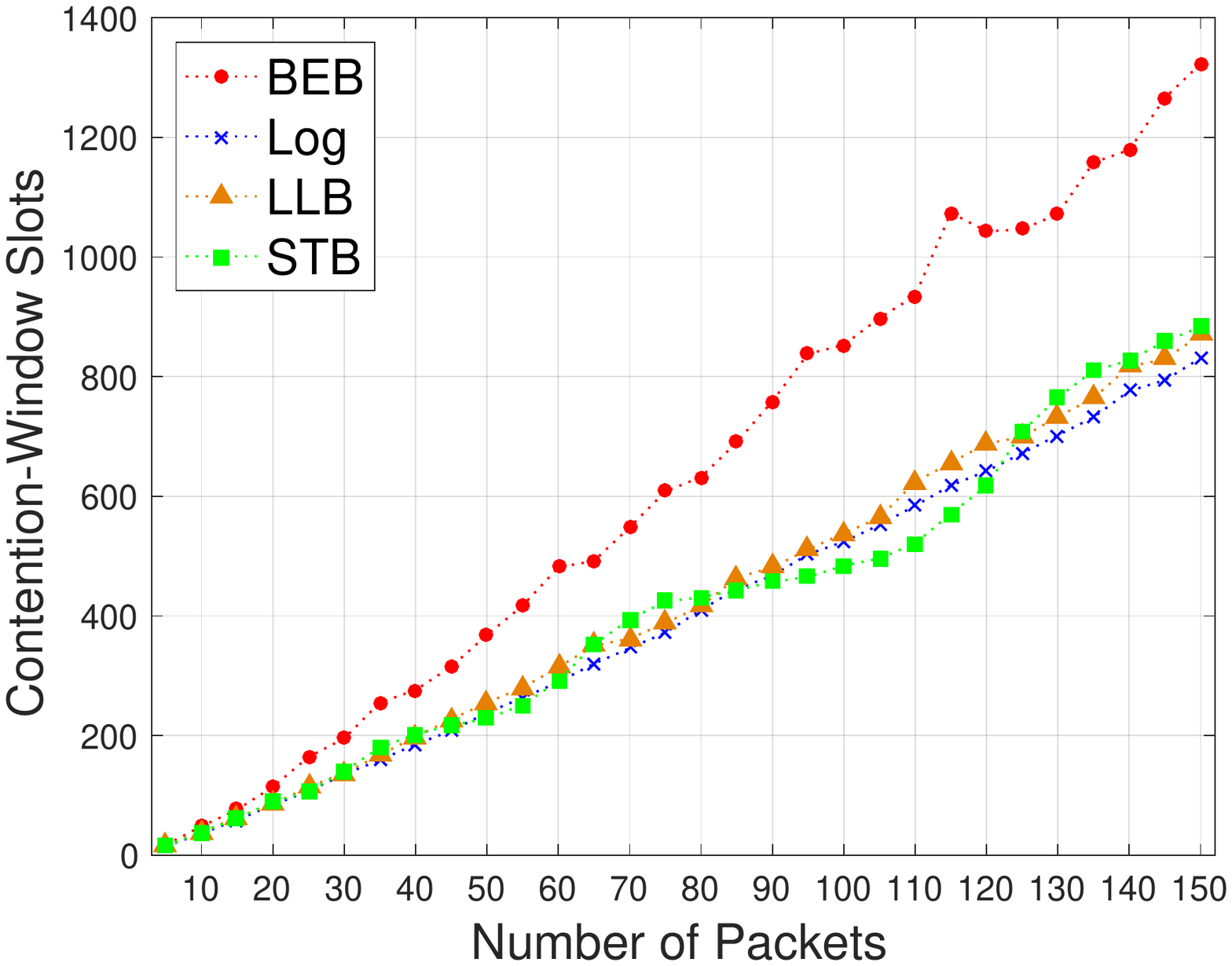}
\vspace{-2.2cm}\caption{\hspace{29pt}(5)}
\end{subfigure}
\begin{subfigure}{0.25\textwidth}
\includegraphics[width=1.2\textwidth]{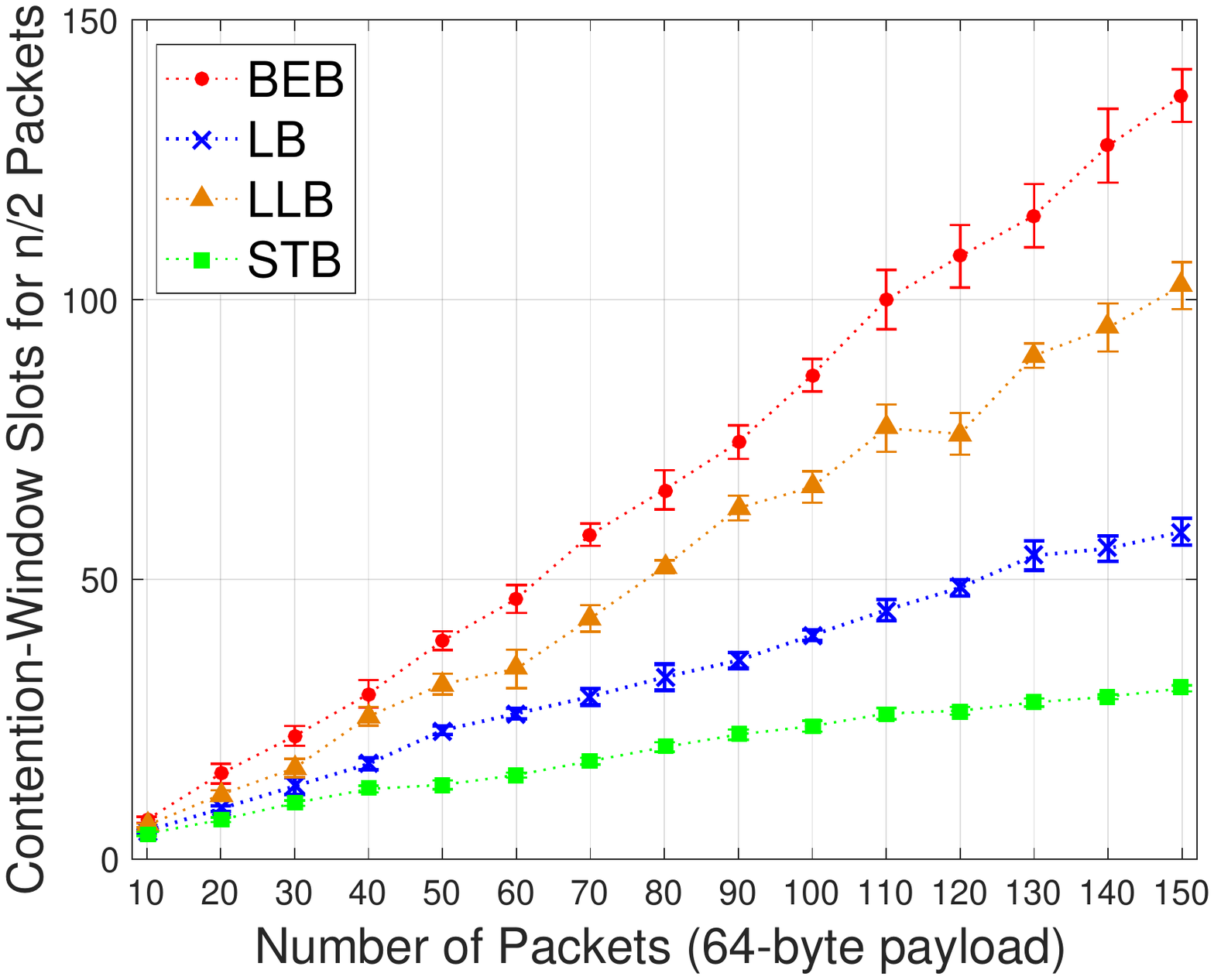} 
\vspace{-2.2cm}\caption{\hspace{29pt}(6)}
\end{subfigure}
\vspace{-4pt}\caption*{{\bf Figures 3-6.} Median values are reported: (3) and (4) CW slots from NS3 experiments with $30$ trials for each value of $n$ with $64$B and $1024$B payloads, respectively, (5) CW slots via Java simulation with $50$ trials for each value of $n$, (6) number of CW slots required to finish $n/2$ packets from NS3 experiments with $20$ trials for each value of $n$. Bars represent $95\%$ confidence intervals.} \label{fig:makespan}
\vspace{-10pt}
\end{figure*}

%%%%%%%%%%%%%%%%%%%%%%%%%%%%%%%%%%%
%%%%%%%%%%%%%%%%%%%%%%%%%%%%%%%%%%%
%%%%%%%%%%%%%%%%%%%%%%%%%%%%%%%%%%%

\noindent{\bf Our Metrics.} For a single batch of $n$ packets, algorithmic results address the number of slots required to complete all $n$ packets. These slots correspond only to those belonging to contention windows, even though many results refer to this as makespan. To avoid confusion, we will refer to this metric more explicitly by \defn{contention-window slots (CW slots)}.

Table~\ref{table:makespan} summarizes the known with-high-probability\footnote{With probability at least $1-1/n^c$ for a tunable constant $c>1$.} guarantees  on CW slots. Note that LB, LLB, and STB each have superior guarantees over BEB, with STB achieving $\Theta(n)$  CW slots which is asymptotically optimal. 

We also make use of a second metric. As described in Section~\ref{sec:802.11}, events occur outside of contention windows (such as SIFS, DIFS, full packet transmission, ACK timeouts). For the duration -- including the time spent in contention windows -- between when the single batch of packets arrives and when the last packet successfully transmits, we refer to \defn{total time}.%\footnote{We exclude the time required for a station to associate with the access point and ARP requests/replies. This would only (unfairly) strengthen the effects we observe.} 
\vspace{-5pt}

%%%%%%%%%%%%%%%%%%%%%%%%%%%%%%%%%%%
%%%%%%%%%%%%%%%%%%%%%%%%%%%%%%%%%%%
%%%%%%%%%%%%%%%%%%%%%%%%%%%%%%%%%%%

\subsection{Theory and Experiment}

We begin by comparing the number of CW slots. The algorithms we investigate are designed to reduce this quantity since all slots in the abstract model occur within some contention window. Under this metric,  LLB, LB, and STB are expected to outperform BEB. 

Throughout, when we report on performance,  we are referring to median values for $n=150$. Percentage increases or decreases are calculated as $100\times(A-B)/B$ where $B$ is always the value for BEB (the ``old'' algorithm) and $A$ corresponds to a value for one of LLB, LB, or STB (the ``new'' algorithms).\vspace{-3pt}

\subsubsection{Contention-Window Slots}

We provide results from our NS3 experiments using both small packets, with a $64$-byte (B)  payload, and large packets, with a $1024$B payload.  %For each value of $n$, we execute $30$ trials and plot the median value along with 95\% confidence intervals.

Figures~3 and 4 illustrate our experimental findings with respect to CW slots.\footnote{The following common approach is used to identify outliers in our data. Let $\Delta$ be the distance between the first and third quartiles. Any data point that falls outside a distance of $1.5 \Delta$ from the median is declared an outlier. We emphasize this results in very few points being discarded; for example, only a {\it single} $n$ value for our $64$B experiments had $5$ outliers (out of 30 trials), and the vast a majority had none.} The behavior generally agrees with theoretical predictions that each of LLB, LB, and STB should outperform BEB. 

Interestingly, LLB incurs a  greater number of CW slots than LB despite despite the former's better asymptotic guarantees. We suspect this is an artifact of hidden constants/scaling and evidence of this is presented later in Section~\ref{sec:large}. 

%; this is confirmed by the Java simulation too.  This behavior persists for $n\leq 1000$ and 

Nevertheless,  LLB, LB, and STB demonstrate improvements over BEB, giving a respective decrease of $49.4$\%, $68.2$\%, and $83.0$\%, respectively, with a $64$B payload. Similarly, LLB, LB, STB demonstrate a respective decrease of $54.2$\%, $69.9$\%,  $84.2$\% with a $1024$B payload.

For comparison, Figure~5 depicts CW slots derived from a simple Java simulation that implements only the assumptions of the abstract model (it ignores wireless effects, details in the protocol stack, etc.). Our NS3 results also roughly agree with this data in terms of magnitude of values and the separation of BEB from the other algorithms; albeit, the performances of LLB, LB, and STB do not separate cleanly in this data.

%What about the distribution of the number of slots required to succeed? 
Finally, Figure~6 presents the number of CW slots (for $64$B) required to complete half the packets, and we make two observations.  First, the remaining $n/2$ packets are responsible for the bulk of the CW slots. Second, the  improvement over~\B~decreases  to $25.0$\%, $56.4$\%,  and $77.7$\% for LLB, LB, and STB, respectively (and similarly for $1024$B).  This difference is due to ``straggling'' packets which survive until relatively large windows are reached. This impacts BEB more than other algorithms given its rapidly increasing window size (and, unlike STB, it does not have a ``backon'' component).\vspace{-8pt}

\begin{figure}[H]
\begin{center}
{
\fbox{\colorbox{light-gray}{
\begin{minipage}[h]{0.42\textwidth} 
\noindent{\bf Result 1.~}{\it Experiment confirms theoretical predictions that LLB, LB, and STB outperform BEB with respect to CW slots.}
\end{minipage}
          }
     }     
}
\end{center}
\vspace{-10pt}
 \end{figure}

\vspace{-8pt}
%%%%%%%%%%%%%%%%%%%%%%%%%%%%%%%%%%%
%%%%%%%%%%%%%%%%%%%%%%%%%%%%%%%%%%%
%%%%%%%%%%%%%%%%%%%%%%%%%%%%%%%%%%%

\begin{figure*}[t]\vspace{-2.4cm}
\captionsetup[subfigure]{labelformat=empty}
\centering
\hspace{-0.7cm}\begin{subfigure}{0.33\textwidth} 
\includegraphics[width=1.2\textwidth]{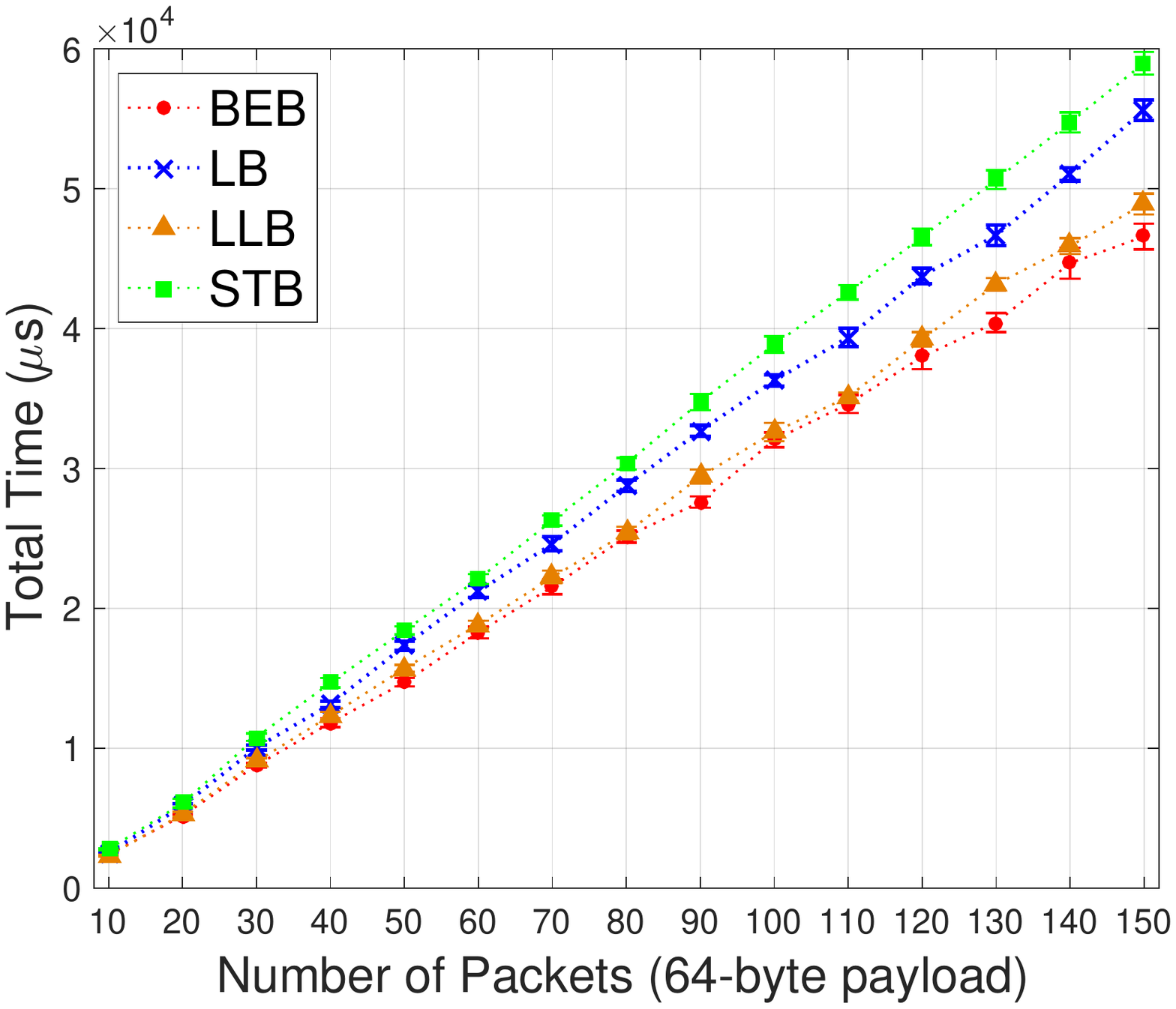} 
\vspace{-2.5cm}\caption{\hspace{35pt}(7)}
\end{subfigure}
\begin{subfigure}{0.33\textwidth} 
\includegraphics[width=1.2\textwidth]{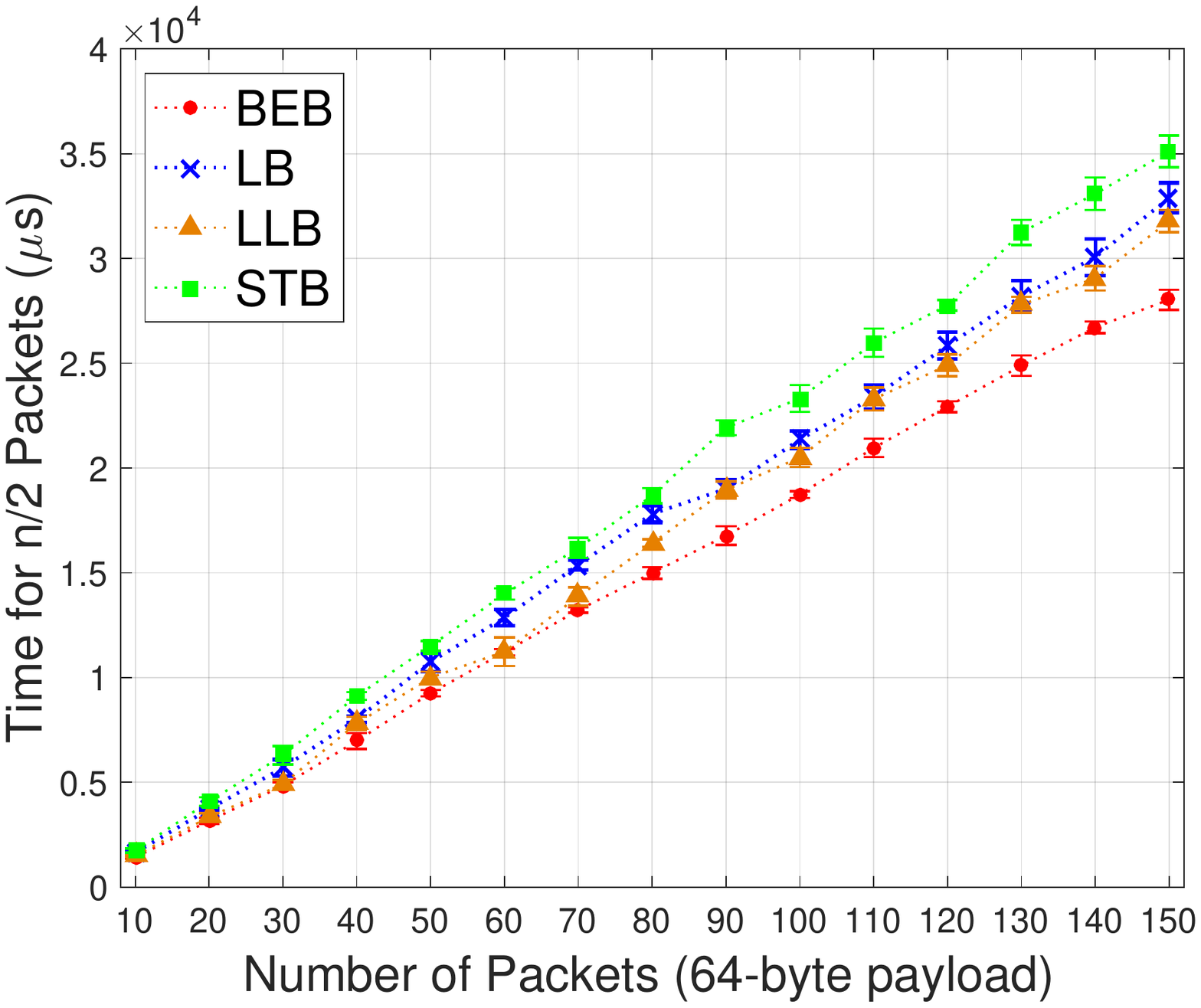} 
\vspace{-2.5cm}\caption{\hspace{35pt}(9)}
\end{subfigure}
\begin{subfigure}{0.33\textwidth}
\includegraphics[width=1.2\textwidth]{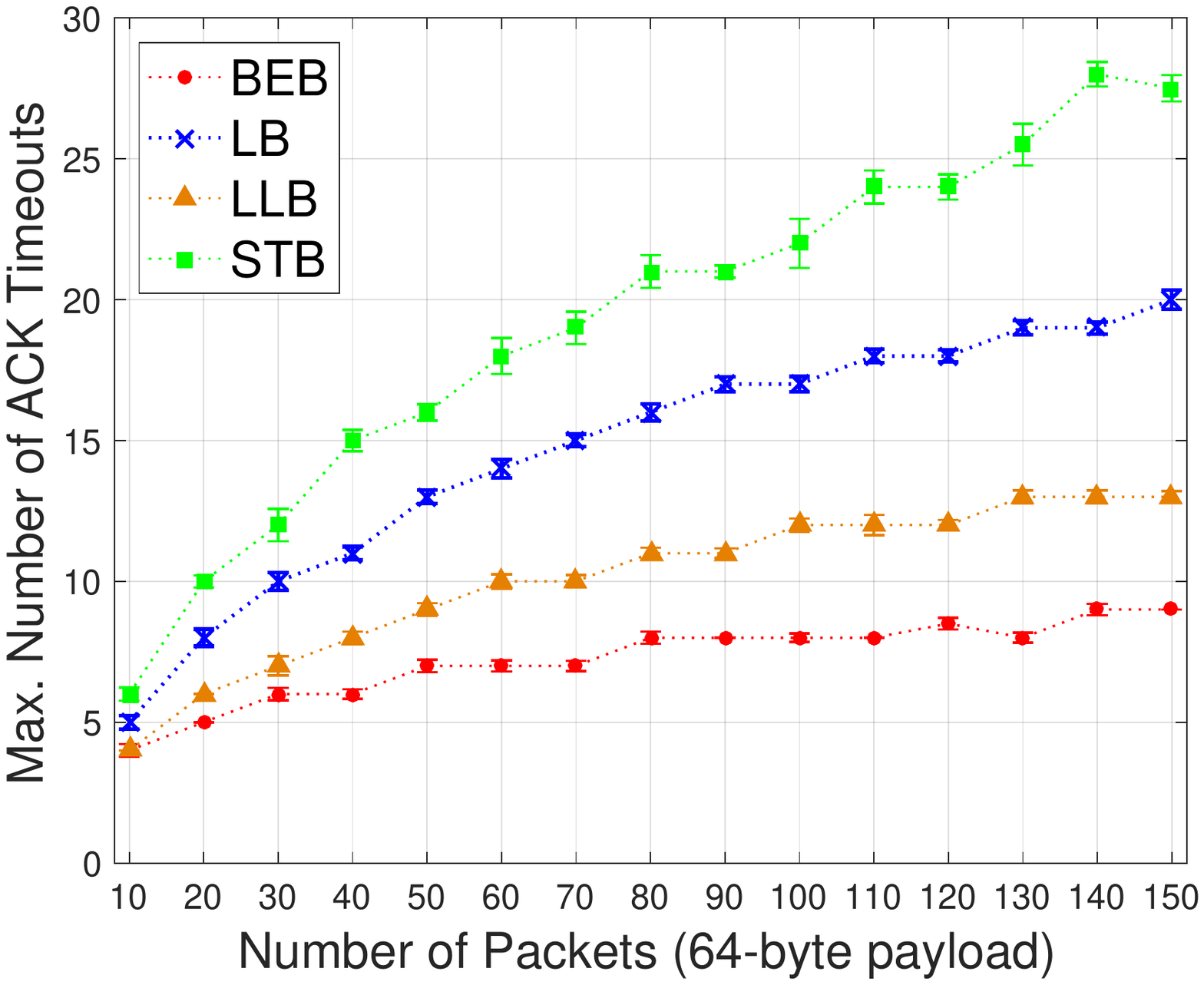}
\vspace{-2.5cm}\caption{\hspace{35pt}(11)}
\end{subfigure}
%\begin{subfigure}{0.25\textwidth}
%\includegraphics[width=1.2\textwidth]{64B-cum-acktime.pdf} 
%\vspace{-2cm}\caption{\hspace{19pt}(10)}
%\end{subfigure} 
\vspace{-0.8cm}
\end{figure*}
%%%%%%%%%%%%%%%%%%%%%%

\begin{figure*}[h]\vspace{-1.8cm}
\captionsetup[subfigure]{labelformat=empty}
\centering
\hspace{-0.7cm}\begin{subfigure}{0.33\textwidth} 
\includegraphics[width=1.2\textwidth]{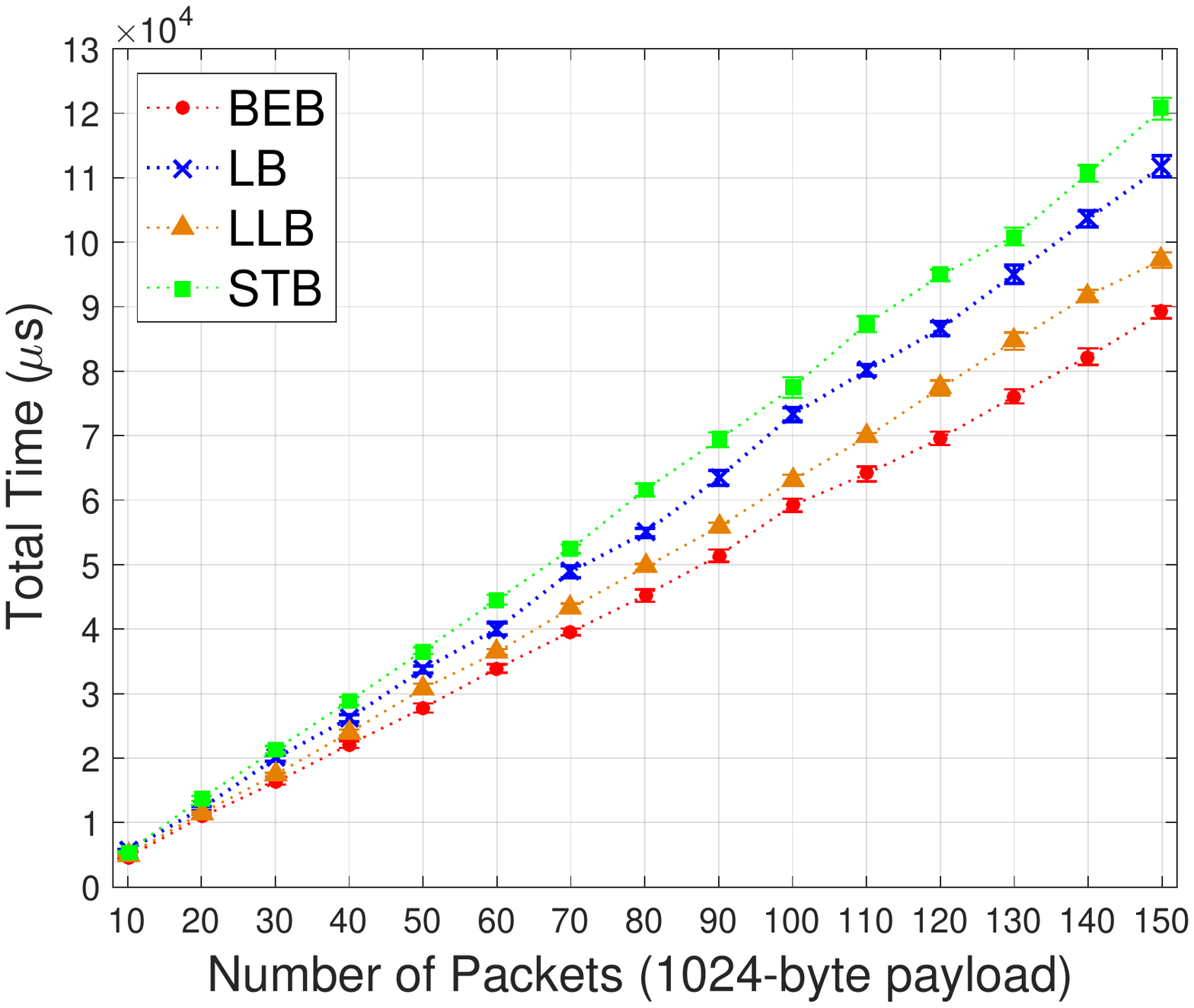} 
\vspace{-2.5cm}\caption{\hspace{35pt}(8)}
\end{subfigure}
\begin{subfigure}{0.33\textwidth} 
\includegraphics[width=1.2\textwidth]{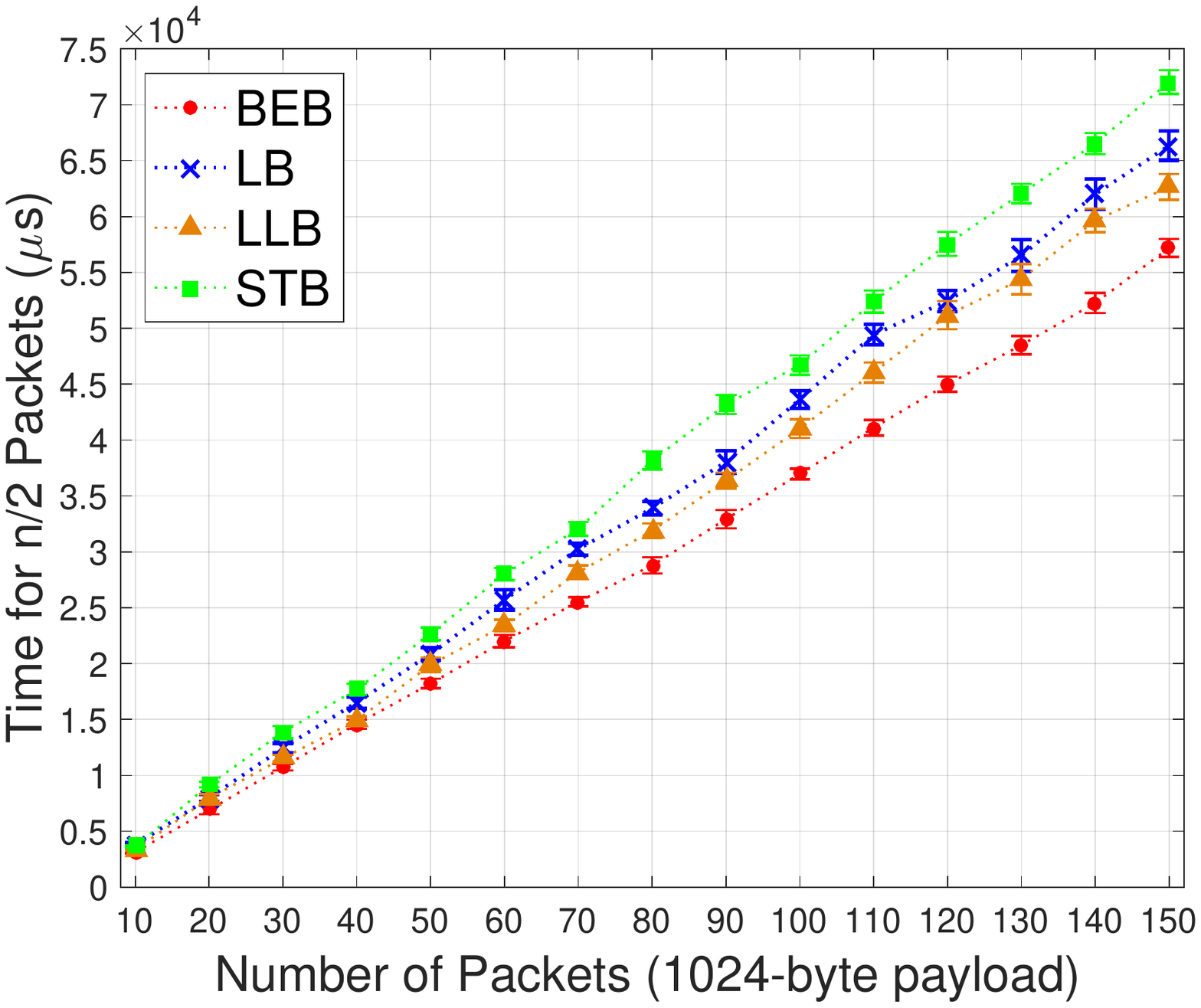} 
\vspace{-2.5cm}\caption{\hspace{35pt}(10)}
\end{subfigure}
\begin{subfigure}{0.33\textwidth}
\includegraphics[width=1.2\textwidth]{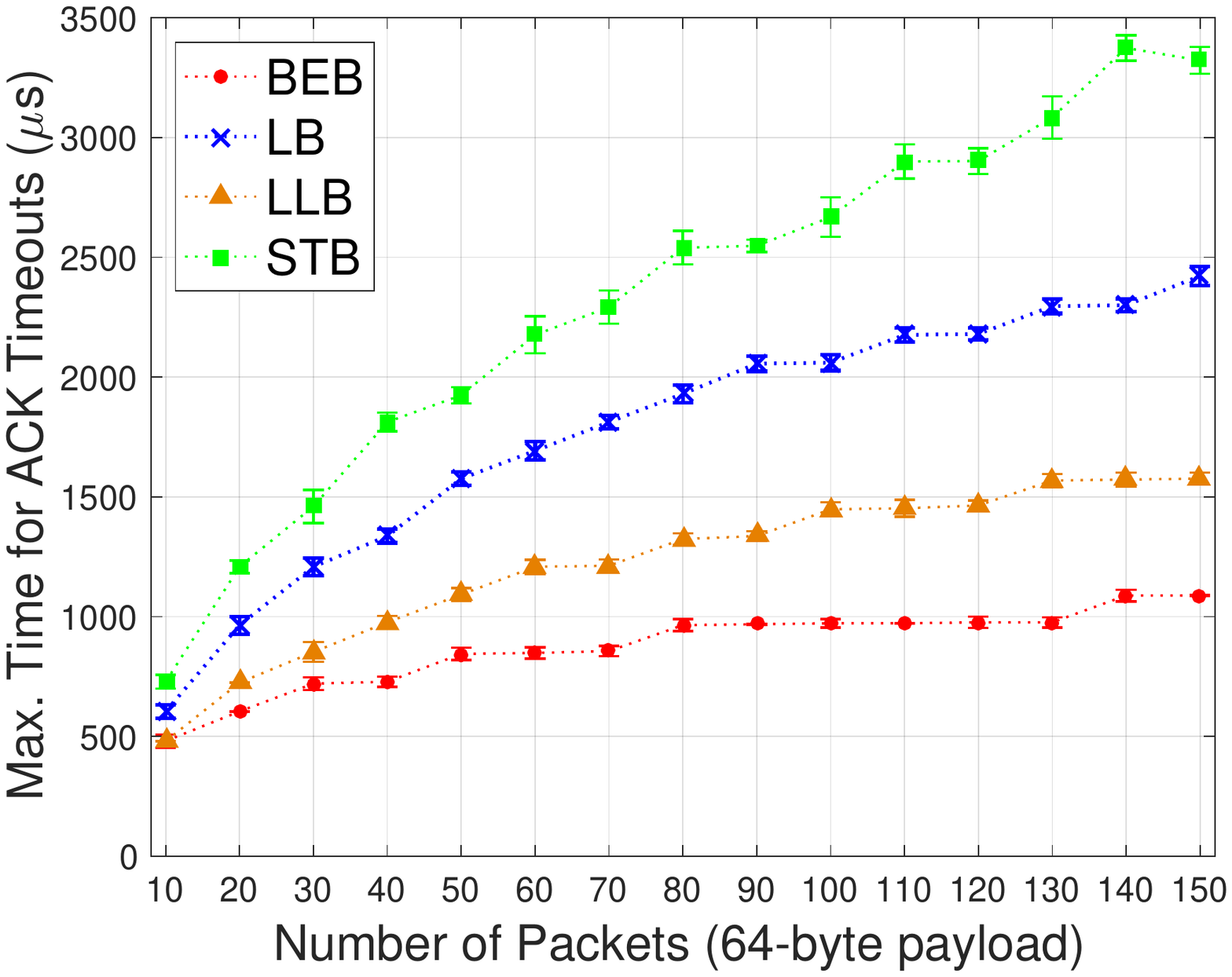}
\vspace{-2.5cm}\caption{\hspace{35pt}(12)}
\end{subfigure}
%\begin{subfigure}{0.25\textwidth}
%\includegraphics[width=1.2\textwidth]{1024B-cum-acktime.pdf} 
%\vspace{-2cm}\caption{\hspace{19pt}(14)}
%\end{subfigure} 
\vspace{-3pt}\caption*{{\bf Figures 7-12.} NS3 results with the median reported from $30$ trials for each value of $n$: (7) and (8) total time for $64$B and $1024$B payloads,  (9) and (10) time required to complete $n/2$ packets with a $64$B payload and $1024$B payloads, (11) the maximum number of ACK timeouts per station over all stations with a $64$B payload, (12) corresponding time spent waiting for ACK timeouts by the station with maximum ACK timeouts. Data plotted with $95\%$ confidence intervals.}\label{fig:total-time}
\end{figure*}

%%%%%%%%%%%%%%%%%%%%%%%%%%%%%%%%%%%
%%%%%%%%%%%%%%%%%%%%%%%%%%%%%%%%%%%
%%%%%%%%%%%%%%%%%%%%%%%%%%%%%%%%%%%

\subsubsection{Total Time}\label{sec:totaltime}

It is tempting to consider the single-batch scenario settled. However, if we focus on the total time for both the $64$B and $1024$B payload sizes, then a different picture emerges.

The degree to which these newer algorithms outperform~\B~is erased as seen from Figures 7 and 8. In fact, the order of performance is reversed with total time ordered from least to greatest as BEB,  LLB, LB,  STB. For $64$B payloads, LLB, LB, and STB suffer an increase of $5.6\%$, $19.3\%$, and $26.5\%$, respectively, over BEB. For $1024$B payloads, the increase is $9.1\%$, $25.4\%$, and $35.4\%$, respectively. Notably, the larger packet size seems to favor BEB. 

What about the time until $n/2$ packets are successfully transmitted? Perhaps newer algorithms do better for the bulk of packets, but suffer from a few stragglers? Interestingly, Figures~9 and 10 suggest that this is not the case. Indeed, for a $64$B payload, BEB performs even better over LLB, LB, and STB with the latter exhibiting an increase of $13.1\%$, $17.3\%$, $25.4\%$, respectively. Similarly, for $1024$B, the percent increase is $10.1\%$, $16.6\%$,  $26.6\%$, respectively.\vspace{-8pt}

%%%%%%%%%%%%%%%%%%%%%%%%%%%%%%%%%%%
%%%%%%%%%%%%%%%%%%%%%%%%%%%%%%%%%%%
%%%%%%%%%%%%%%%%%%%%%%%%%%%%%%%%%%%

\begin{figure}[H]
\begin{center}
{
\fbox{\colorbox{light-gray}{
\begin{minipage}[h]{0.42\textwidth} 
\noindent{\bf Result 2.~}{\it In comparison to BEB, the total time for each of LLB, LB, and STB is significantly worse.}
\end{minipage}
          }
     }     
}
\end{center}
\vspace{-10pt}
 \end{figure}

These findings are troubling since, arguably, total time is a more important performance metric in practice than CW slots. Critically, we note that this behavior is detected only through the use of NS3; it is not apparent from the simpler Java simulation. What is the cause of this phenomenon?

%%%%%%%%%%%%%%%%%%%%%%%%%%%%%%%%%%%
%%%%%%%%%%%%%%%%%%%%%%%%%%%%%%%%%%%
%%%%%%%%%%%%%%%%%%%%%%%%%%%%%%%%%%%

\subsection{The Cost of Collisions}\label{sec:hidden}

The number of ACK timeouts per station provides an important hint. As Figure~8 shows, the newer algorithms are incurring substantially more ACK timeouts which, in turn, corresponds to more (re)transmissions.

This evidence points to collisions as the main culprit. In particular, the way in which collision detection is performed means that each collision is costly in terms of time. In support of our claim, we decompose this delay into three portions using BEB (for $n=150$) as an example throughout:\smallskip

\noindent{\defn{(I) Transmission Time}.} (Re)transmissions are expensive. A packet of size $128$B ($64$B payload plus $64$B overhead) requires roughly $19 \mu$s plus the associated $20\mu$s preamble.   The maximum number of ACK timeouts for BEB -- and, thus, the number of collisions -- experienced by an unlucky station is $9$. 

Most collisions should involve only a handful of stations given the growth of CWs. (Why? Recall from Figure~6 that $n/2$ packets require the vast majority of CW slots to finish. These $n/2$ packets succeed only in larger windows --  of size roughly equal to or greater than $n$ for BEB -- and do not finish immediately due to collisions, each of which should involve only a few stations given such window sizes). If two stations are involved in each collision, this results in $75(9/2)$ non-overlapping (or \defn{disjoint}) collisions, for an aggregate  duration of roughly $75(9/2)(19\mu s + 20\mu s) = 13,163\mu s$.\footnote{We do not add the time for the final/successful transmissions (this would only increase the value); our focus is on the transmissions associated with collisions.} \smallskip  

\noindent{\defn{(II) ACK Timeouts}.} Given a collision, the AP fails to obtain the transmission and the corresponding stations incur an ACK timeout before concluding that a collision occurred. This delay is significant -- roughly $1,100\mu$s for BEB with $n=150$ (see Figure 12) -- but an order-of-magnitude less than the transmission time. \smallskip

\noindent{\defn{(III) Contention-Window Slots}.} BEB incurs $886$ CW slots for $n=150$, each of duration $9\mu$s, spent in CWs which yields $7,974\mu s$. \smallskip

For BEB at $n=150$, these three values yield a very conservative lower bound on the total time of $22,237\mu s$; for instance, we have not accounted for the SIFS and DIFS. This back-of-the-envelope calculation conforms to the magnitude of values observed in Figure~7, and it highlights two important facts. First, both transmission time and CW slots contribute significantly to total time, with ACK timeouts being a distant third. Second, collisions greatly impact total time -- far more than CW slots -- by forcing retransmissions.

Underlining this second point, we note that the transmission time for the $1024$B payload is larger at roughly $75(9/2)(161\mu s + 20\mu s) = 61,088\mu s$.\footnote{The number of ACK timeouts for  $1024$B is roughly the same, even though packet size has increased. This aligns with our findings in Section 4.} By comparison, the CW slots contribute roughly $973\times 9\mu s = 8,757\mu s$. \smallskip

\noindent{\bf ACK Timeout \boldmath{$\approx$} Collision.} It is true that not all ACK timeouts necessarily imply that the corresponding packet suffered a collision. For example, an ACK might be lost due to wireless effects even if the packet was transmitted without any collision. Note that, in such a case, the sending station still diagnoses that a collision has occurred and so the same costs described in (I)-(III) hold. 

For our simple NS3 setup, virtually all ACK failures result from a collision. This is evident from Figure~13 which illustrates a trial with $n=20$ under BEB using a $64$B payload. Collisions occur only when two or more stations transmit (duration of transmission denoted by a thick blue line) at the same time and the result is an ACK timeout event (indicated by a thin red line); in all other cases, the transmission is successful and the corresponding ACK is received.\smallskip

\begin{figure}[t]
\vspace{-3.4cm}
\begin{center}
\mbox{\hspace{-1cm}}\includegraphics[width=0.55\textwidth]{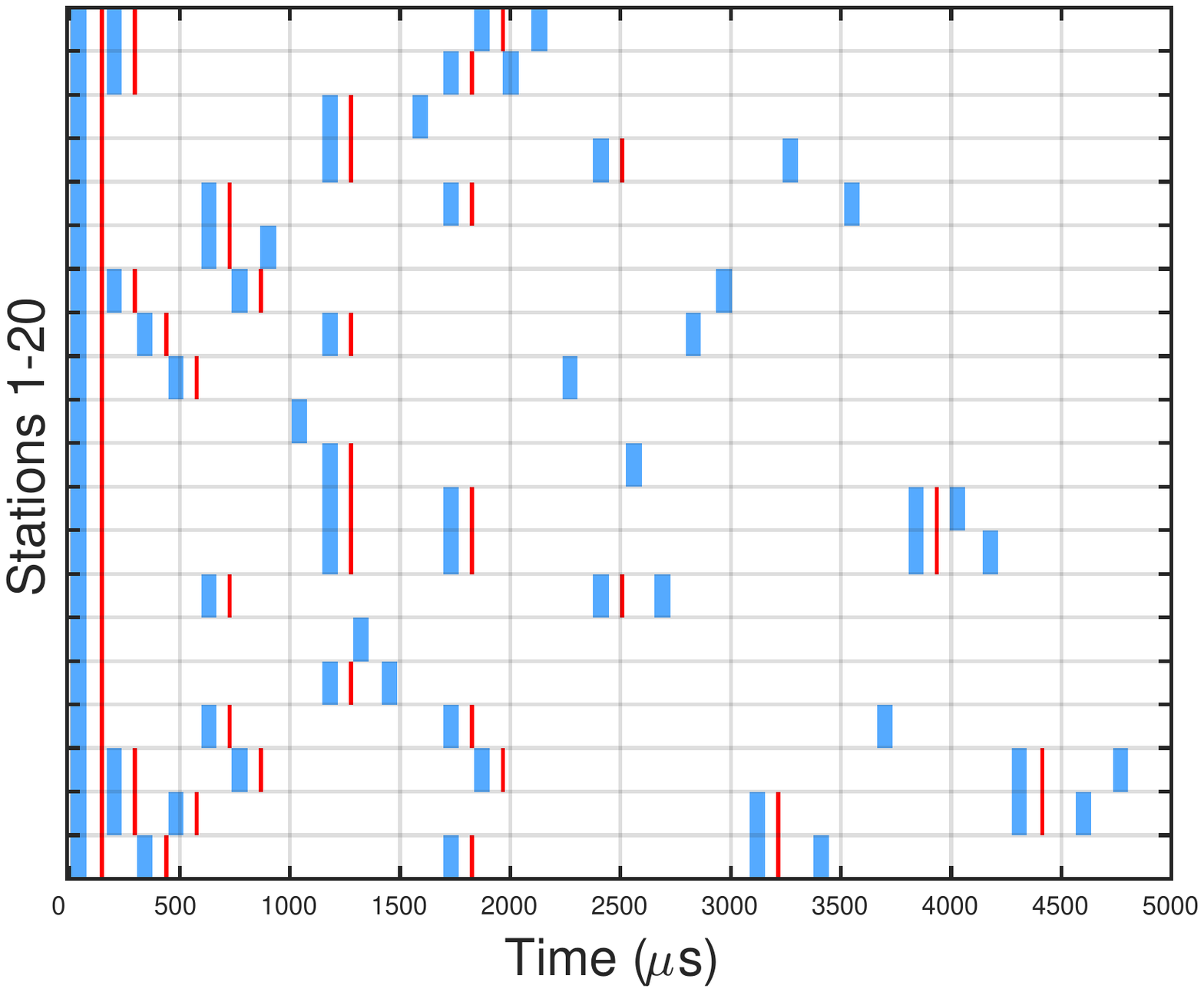}
\vspace{-3.6cm}
\caption*{ {\bf Figure 13.} Execution of BEB with $20$ stations.}
\end{center}
\vspace{-0.5cm}
\end{figure}

\noindent{\bf Disjoint Collisions.} We observe that the total time does not grow linearly with the maximum number of ACK timeouts -- equivalently, collisions -- experienced by a station. Under LB, an unlucky station suffers roughly twice the number of ACK timeouts, but the total time of LB is not twice that of BEB. 

Why? Consider $n$ stations where each collision involves only two stations. Then, there are $n/2$ disjoint collisions and each is added to the total time.  In contrast, consider the opposite extreme where all $n$ stations transmit at the same time and collide. Then, there is one collision which adds only a single failed transmission time to the total time.

The number of stations involved in a single collision is larger for algorithms whose CWs grow more slowly, such as LB and LLB.  For STB, a similar phenomenon is at work; the backon component yields collisions involving many stations. That is, LB, LLB, and STB are closer to the second case, while BEB is closer to the first.
 
From our experimental results, we see assumption A2 is not accurate with regards to the cost of failure:\vspace{-10pt}
\begin{figure}[H]
\begin{center}
{
\fbox{\colorbox{light-gray}{
\begin{minipage}[h]{0.42\textwidth} 
\noindent{\bf Result 3.~} {\it The impact from collision detection is not properly accounted for by A2. This impact is primarily a function of:\vspace{-0pt}
\begin{itemize} 
\item transmission time, and \vspace{-0pt}
\item time spent in contention windows,\vspace{-0pt}
\end{itemize}
\noindent with the former dominating.}% at this scale.}
%while delay from ACK timeouts is secondary.
\end{minipage}
          }
     }     
}
\end{center}
\vspace{-15pt}
 \end{figure}

\noindent{\bf RTS/CTS.} Although it is not examined in detail in our work, we briefly remark on the use of  RTS/CTS. When enabled, stations can experience collisions among the RTS frames (instead of the packets). These are smaller in size ($20$B), but the remainder of the total-time calculation remains the same, and additional time is incurred due to additional inter-frame spaces and the transmission of CTS frames. 

For very large packets, we may expect RTS/CTS to mitigate the transmission-time cost, but not for small to medium-sized packets where the overhead from this mechanism might even cause worse performance. Ultimately, we observe the same qualitative behavior when RTS/CTS is enabled. For example, without RTS/CTS, recall from Section~\ref{sec:totaltime} that the total time for LLB (BEB's closest competitor) increases by $5.6$\% and $9.1$\% for the $64$B and $1024$B, respectively, over BEB. With RTS/CTS, the  increases are $10.7$\% and $7.5$\%.

%%%%%%%%%%%%%%%%%%%%%%%%%%%%%%%%%%
%%%%%%%%%%%%%%%%%%%%%%%%%%%%%%%%%%
%%%%%%%%%%%%%%%%%%%%%%%%%%%%%%%%%%

\subsubsection{Backing Off Slowly is Bad}

The reason for the discrepancy between theory and experiment is now apparent. LLB increases each successive contention window by a smaller amount than BEB; in other words, LLB is {\it backing off more slowly}; the same is true of LB. Informally, this slower-backoff behavior is the reason behind the superior number of CW slots for LB and LLB since they linger in CWs where the contention is ``just right'' for a significant fraction of the packets to succeed. However, backing off slowly also inflicts a greater number of collisions.

Note that BEB backs off faster, jumping away from such favorable contention windows and thus incurring many empty slots. This is undesirable from the perspective of optimizing the number of CW slots. However, the result is fewer collisions. Given the empirical results, this appears to be a favorable tradeoff.

We explicitly note that LLB backs off faster than LB. In this way, LLB is ``closer'' to BEB and, therefore, is not outperformed as badly as illustrated in Figures 7 and 8.

\begin{figure}[H]
\begin{center}
{
\fbox{\colorbox{light-gray}{
\begin{minipage}[h]{0.42\textwidth} 
\noindent{\bf Result 4.~} {\it For algorithm design, optimizing CW slots at the expense of increased collisions is a poor design choice.}
\end{minipage}
          }
     }     
}
\end{center}
\vspace{-8pt}
 \end{figure}

Can we quantify the tradeoff between CW slots and collisions? From our discussion in Section~\ref{sec:hidden}, the total time for an algorithm $A$, denoted by $\mathcal{T}_A$, is approximated as: \vspace{-6pt}

$$\mathcal{T}_A = \mathcal{C}_A\cdot (P + \rho)  +  \mathcal{W}_A\cdot{s} $$

\noindent where  $\mathcal{C}_A$ is the number of disjoint collisions, $P$ is the transmission time for a packet, $\rho$ is the preamble duration, $\mathcal{W}_A$ is the corresponding number of CW slots, and $s$ is the duration of a slot.

Abstracting further, we may treat $\rho$ and $s$ as constants to get:\vspace{-8pt} %.\footnote{Recall that the abstract model typically assumes a packet fits within a slot, but from Section~\ref{sec:802.11}, we know that a slot is only the beginning of the data transmission.} This yields:\vspace{-15pt}

$$\mathcal{T}_A =  \Theta\left(\mathcal{C}_A\cdot P  +  \mathcal{W}_A\right)$$

In other words, total time depends on the number of disjoint collisions (which depends on $n$) --- each of which has a severity that depends on $P$ --- and the number of CW slots (which depends on $n$). 

How does $P$ behave? We assume it is proportional to packet size. For small values of $n$, it seems reasonable to consider $P=\Theta(1)$.  However, if we are interested in the asymptotic behavior of $\mathcal{T}_A$, then $P$ should not be treated as a constant. Arguably, as $n$ scales, the number of bits required to address devices must also increase, and it is not uncommon to assume $P$ scales as the logarithm of $n$.  

Previous results have already established $\mathcal{W}_A$, so the  parameter of interest is $\mathcal{C}_A$, which we investigate next.\vspace{-3pt}

%%%%%%%%%%%%%%%%%%%%%%%%%%%%%%%%%%%
%%%%%%%%%%%%%%%%%%%%%%%%%%%%%%%%%%%
%%%%%%%%%%%%%%%%%%%%%%%%%%%%%%%%%%%

\section{Bounds on Collisions}\label{sec:theory}

In order to provide additional support for our empirical findings, we derive asymptotic bounds on $\mathcal{C}_A$.  In comparison to BEB, we demonstrate that STB is asymptotically equal while both LLB and LB suffer from asymptotically more disjoint collisions.  

Our arguments are couched in terms of packets and slots, but what follows is a balls-into-bins analysis. To bound $\mathcal{C}_A$, we are interested in the number of bins (where bins make up the slots in a CW) that contain two or more balls; this is a disjoint collision (or just a {\it collision}). The second column of Table~3 presents our results.

\subsection{Upper Bounding Collisions in BEB}

\begin{claim}
For a single batch of $n$ packets, with high probability the number of collisions for \B~is $O(n)$.\vspace{-5pt}
\end{claim}
\begin{proof}
In the execution of~\B, consider a contention window of size $n 2^i$ for an integer $i\geq 0$; let the windows be indexed by $i$. Note that up to window $i=0$, we have $O(n)$ collisions since there are $O(n)$ slots by the sum of a geometric series. 

Let the indicator random variable $X_j=1$ if slot $j$ in window $i$  is a collision; $X_j=0$ otherwise. We map this to a balls-and-bins problem, where a ball corresponds to a packet and a bin corresponds to a slot in a contention window. %Since the number of bins is at least the number of balls, for window $i$:\vspace{-3pt}
\begin{eqnarray*}
Pr[X_j=1]  & = &  O\left( \binom{n }{2}\hspace{0pt}\left(\frac{1}{n2^i}\right)^2\hspace{-3pt} \left(1- \frac{1}{n2^i}\right)^{n-2}\right)\\
& = & O\left(\frac{1}{2^{2i}}\right)
\end{eqnarray*}
Pessimistically, assume $n$ balls are dropped in each consecutive window $i$; in actuality, packets finish over these windows and reduce the probability of collisions. Let $L_i=\sum_{j=1}^{n2^i} X_j$ be the number of collisions in window $i$. By linearity of expectation:\vspace{-4pt}
$$E[L_i] = \sum_{j=1}^{n2^i} E\left[X_j\right] = O\left(\frac{n}{2^{i}}\right)$$
\noindent Using the method of bounded differences~\cite{dubhashi:concentration}, w.h.p. $L_i$ is tightly bounded to its expectation. By~\cite{BenderFaHe05}, w.h.p. \B~finishes within $m=O(\lg\lg n)$ windows  and so $\mathcal{C}_{BEB} = \sum_{i=0}^{m} L_i = O(n/2^i) = O(n)$. 
\end{proof}

 %and so the smallest summand is $L_{m} = O(n/\texttt{polylog}(n))$. It follows that 

%%%%%%%%%%%%%%%%%%%%%%%%%%%%%%%%%%%
%%%%%%%%%%%%%%%%%%%%%%%%%%%%%%%%%%%
%%%%%%%%%%%%%%%%%%%%%%%%%%%%%%%%%%%

\subsection{Lower Bounding Collisions in LLB and LB}
 
The specification of LLB analyzed here is slightly different from the description in Section~\ref{sec:single-batch}; the contention window size doubles, but each such window of size $w$ is repeated for $\lg\lg w$ iterations. With respect to CW slots and disjoint collisions, this is asymptotically equivalent to the strictly monotonic version described earlier~\cite{BenderFaHe05}. 

The window size of interest is $\Theta(n/\lg\lg\lg n)$, since LLB finishes within a window of this magnitude~\cite{BenderFaHe05}. We first prove an upper bound of $o(n)$ successes  in a single execution of this window. This allows us to claim $\Omega(\lg\lg n)$ iterations exist where $\Theta(n)$ packets remain unfinished. Next, we prove that for each such iteration, $\Omega(n/\lg\lg\lg n)$ collisions occur yielding a total of $\Omega(n\lg\lg n/\lg\lg\lg n)$ collisions. \vspace{-2pt}

\begin{lemma}\label{lem:LLB-succ}
Assume $\epsilon n$ packets and a CW of size  $cn/\lg\lg\lg n$ for a sufficiently large constant $\epsilon\leq 1$ and sufficiently small constant $c>0$. With high probability, at most $O(n/(\lg\lg n)^{d})$ packets succeed in the CW for a constant $d>1$ depending on $\epsilon$ and $c$.\vspace{-3pt}
\end{lemma}
\begin{proof}
Let $Y_j=1$ if a packet succeeds in slot $j$, otherwise $Y_j=0$. We have: \vspace{-5pt}
\begin{eqnarray*}
Pr[Y_j=1] & = & \binom{\epsilon n}{1}\left(\frac{\lg\lg\lg n}{cn}\right) \hspace{-3pt} \left(1- \frac{\lg\lg\lg n}{cn}\right)^{\epsilon n-1}\\
&\leq & \frac{\epsilon\lg\lg\lg n}{c(\lg\lg n)^{\epsilon\lg(e)/c}} \leq O\left(\frac{1}{c(\lg\lg n)^{d}}\right) 
\end{eqnarray*}
\noindent for a constant $d>1$ where the last line follows from noting that $\epsilon\lg(e)/c > 1$ for a sufficiently large constant $\epsilon$ and a sufficiently small constant $c>0$.  Let $Y=\sum_j Y_j$, then:\vspace{-5pt}
$$E[Y] = O\left(\frac{n}{(\lg\lg n)^{d}}\right)$$
\noindent By the method of bounded differences, w.h.p. this is tight.
\end{proof}

\begin{claim}\label{claim:LLB}
For a single batch of $n$ packets, with high probability~\LLB~experiences $\Omega\left(\frac{n\lg\lg n}{\lg\lg\lg n}\right)$ collisions.\vspace{-5pt}
\end{claim}
\begin{proof}
We focus on a contention window $w$ of size $cn/\lg\lg\lg n$ for a sufficiently small constant $c>0$. Conservatively, we do not count collisions prior to this window (counting these can only improve our result).

Prior to this window, w.h.p. $o(n)$ packets have succeeded. To see this, note that w.h.p. no packet finishes prior to a window of size $\Theta(n/\lg n)$. The number of intervening windows before reaching size $cn/\lg\lg\lg n$ is less than $\lg\lg n$. Pessimistically assume each intervening window has size $cn/\lg\lg\lg n$, then each results in  $O(n/(\lg\lg n)^{d})$ successful packets  w.h.p. by Lemma~\ref{lem:LLB-succ} for $d>1$. Each such intervening window executes $O(\lg\lg n)$ times. Therefore, the total number of packets finished is still $O(n/(\lg\lg n)^{d'})$ for a constant $d'$ depending only $d$, and so $\Omega(n)$ packets remain.

 %there are fewer than $\lg n$ prior windows -- each  -- over which w.h.p. $O(n/(\lg\lg n)^{\epsilon'\lg(e)/c})$ packets succeed by Lemma~\ref{lem:LLB-succ}. For sufficiently small $c$, this leaves 

In this window, assume that $\epsilon n$ packets exist for some constant $\epsilon>0$. Let $X_j=1$ if slot $j$ contains a collision; otherwise, $X_j=0$. Then we have $Pr[X_j=1]$:\vspace{-5pt}
\begin{eqnarray*}
& = &  \hspace{-7pt}1 - \sum_{k=0}^1 \binom{\epsilon n}{k}\hspace{0pt}\left(\frac{\lg\lg\lg n}{cn}\right)^k \hspace{-3pt} \left(1- \frac{\lg\lg\lg n}{cn}\right)^{\epsilon n-k}\\
& \geq & 1 - \left( 1-\frac{\lg\lg\lg n}{cn}\right)^{\epsilon n} - \frac{\epsilon \lg\lg\lg n}{c}\left(1-\frac{\lg\lg\lg n}{cn} \right)^{\epsilon n - 1 }\\
%& \geq & 1-  \left( 1-\frac{\lg\lg\lg n}{cn}\right)^{\epsilon n}  \left( 1 + \frac{2\epsilon\lg\lg\lg n}{c}\right)\\
& \geq & 1 - O \left( \frac{\epsilon\lg\lg\lg n}{ c\left( \lg\lg n \right)^{2\epsilon\lg(e)/c}} \right)  \mbox{by Taylor series}\\
& = & \Omega(1)
\end{eqnarray*}
\noindent Let $X = \sum_j X_j$. The expected number of collisions over the contention window $w$ is:\vspace{-5pt}
\begin{eqnarray*}
E[X] &=& \sum_j E[X_j] = \Omega\left(\frac{cn}{\lg\lg\lg n}\right)\vspace{-5pt}
\end{eqnarray*}
\vspace{-7pt}

\noindent By the method of bounded differences, w.h.p. this is tight.

The window $w$ is executed $\lg\lg(cn/\lg\lg\lg n) = \Omega(\lg\lg n)$ times. By Lemma~\ref{lem:LLB-succ}, $O(n/(\log\log n)^d)$ packets are succeeding in each such execution for some constant $d>1$ given $c$ is sufficiently small. Thus, there will be at least $\epsilon n$ packets remaining in each execution for some sufficiently small constant $\epsilon > 0$. By the above lower bound on the number of collisions, w.h.p. this results in $\mathcal{C}_{LLB} = \Omega(cn\lg\lg n/\lg\lg\lg n)$ collisions.
\end{proof}

\noindent An argument similar to that used to support  Claim~\ref{claim:LLB} yields:

\begin{claim}
For a single batch of $n$ packets, with high probability~\LB~experiences $\Omega\left(\frac{n\lg n}{\lg\lg n}\right)$ collisions.
\end{claim}

%%%%%%%%%%%%%%%%%%%%%%%%%%%%%%%%%%%
%%%%%%%%%%%%%%%%%%%%%%%%%%%%%%%%%%%
%%%%%%%%%%%%%%%%%%%%%%%%%%%%%%%%%%%

\begin{table}[t] 
\begin{center}
{
\begin{tabular}{ |p{1.8cm}|c|c|  }
\hline
\rowcolor{LightCyan}\hspace{0pt}{\bf Algorithm $A$} &  {\bf Num. of Collisions $\mathcal{C}_A$}  & {\bf Total Time $\mathcal{T}_A$}\\
\hline
\hspace{20pt}BEB & $O(n)$ &  $O\left(n\cdot P + n\log n\right)$\\
\hline
\hspace{20pt}LB& $\Theta\left( \frac{n\log n}{\log\log n} \right)$ & $ \Omega\left(\frac{n\log n}{\log\log n}\cdot P\right)$\\
\hline
\hspace{20pt}LLB & $\Theta\left( \frac{n\log\log n}{\log\log\log n} \right)$  & $\Omega\left( \frac{n\log\log n}{\log\log\log n}\cdot P \right)$\\
\hline
\hspace{20pt}STB & $\Theta\left(n\right)$ & $\Theta\left(n\cdot P\right)$\\
\hline
\end{tabular}
}\vspace{-0pt}\caption{Asymptotic bounds on collisions and total time.}% for a single batch of $n$ packets.}\label{table:colls}
\vspace{-10pt}
\end{center}
\end{table}

%%%%%%%%%%%%%%%%%%%%%%%%%%%%%%%%%%%
%%%%%%%%%%%%%%%%%%%%%%%%%%%%%%%%%%%
%%%%%%%%%%%%%%%%%%%%%%%%%%%%%%%%%%%

\begin{figure}[t]
\centering\vspace{-40pt}
\includegraphics[width=0.5\textwidth]{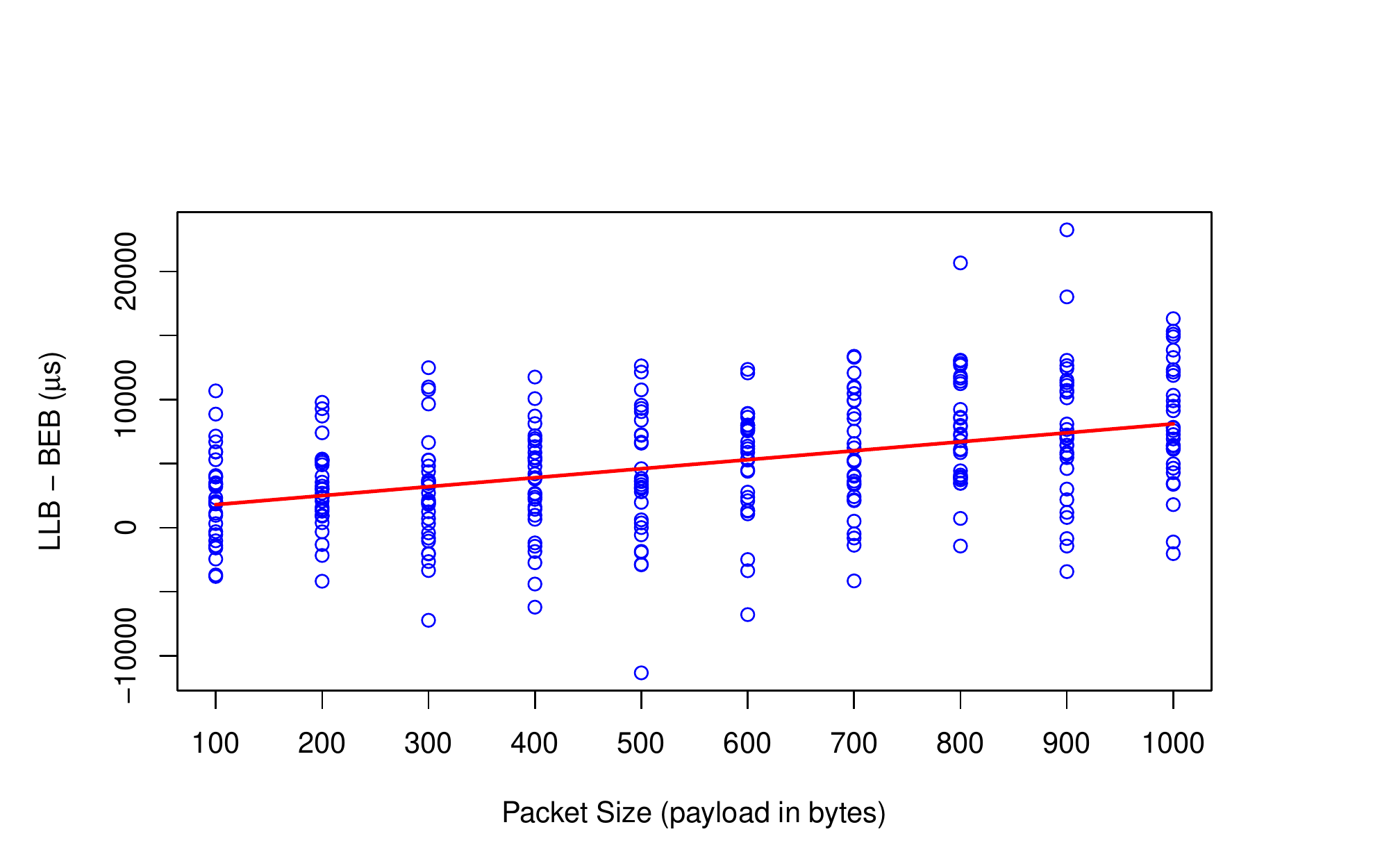}\vspace{-5pt}
\caption*{ {\bf Figure 14.} The difference in total time between LLB and BEB for $n=150$ as packet size increases (30 trials per size).}\vspace{-10pt}
\end{figure}

\subsection{Upper Bounding Collisions in STB}

For a single batch of $n$ packets, it is known that w.h.p. STB~has $\mathcal{W}_{STB} = O(n)$ and this is a trivial upper bound on the number of collisions. A straightforward argument allows us to derive a lower bound of $\Omega(n)$. \vspace{-5pt}

\begin{claim}
For a single batch of $n$ packets, with high probability~\STB~experiences $\Omega(n)$ collisions.
\end{claim}
\begin{proof}
Consider a window of size $n/8$. The total number of slots up to the end of this window (including all the backon windows) is less than $n/2$; therefore, more than $n/2$ packets have not finished by this point. In the next window, which has size $n/4$, the probability of a collision is constant. Therefore, the expected number of collisions is $\mathcal{C}_{STB}=\Omega(n)$ and this is tight by the method of bounded differences.
\end{proof}

Although BEB and STB are asymptotically equal in the number of collisions suffered, we expect the hidden constant in the big-O notation for STB to be larger due to the backon component. We consider this question, amongst others involving asymptotic performance, later in Section~\ref{sec:large}.

%%%%%%%%%%%%%%%%%%%%%%%%%%%%%%%%%%%
%%%%%%%%%%%%%%%%%%%%%%%%%%%%%%%%%%%
%%%%%%%%%%%%%%%%%%%%%%%%%%%%%%%%%%%

\subsection{Asymptotic Behavior of Total Time}\label{sec:as-total}

Plugging in the results from Section~\ref{sec:theory}, our formula for $\mathcal{T}_A =  \Theta\left(\mathcal{C}_A\cdot P  +  \mathcal{W}_A\right)$ yields the third column Table~3. 

Recall that for small $n$, treating  $P$ as a constant is reasonable. However, for large values of $n$, one may argue that $P$ ought to be treated as a slowly growing function of $n$, such as $\Omega(\log n)$.  In this case, we note that both $\mathcal{T}_{LB}$ and $\mathcal{T}_{LLB}$ exceed $\mathcal{T}_{BEB}$ asymptotically. In fact, even a smaller bound $P = \omega(\lg n \lg\lg\lg n/\lg\lg n)$ is sufficient to yield this asymptotic behavior.\vspace{-8pt}

\begin{figure}[H]
\begin{center}
{
\fbox{\colorbox{light-gray}{
\begin{minipage}[h]{0.42\textwidth} 
\noindent{\bf Result 5.~} {\it Theoretical bounds on total time imply that LLB and LB should underperform both BEB and STB for sufficiently large $n$ and $P$.} 
\end{minipage}
          }
     }     
}
\end{center}
\vspace{-10pt}
 \end{figure}

This analysis offers support for our conjecture that the number of collisions is an important metric -- perhaps more so than the number of CW slots -- when it comes to the design of contention-resolution algorithms.

With respect to total time, recall that in Section~\ref{sec:totaltime} an increase in packet size was seen to favor BEB over LLB, the latter being the closest competitor to BEB in our experiments (although, STB is asymptotically superior and we address this issue in Section~\ref{sec:large}). This aligns with the above discussion. Furthermore, as empirical support for our claim, we use NS3 to examine the relative performance of these two algorithms as packet size increases in Figure~14.

As the packet size grows, LLB performs increasingly worse than BEB. We fit a linear regression model of LLB - BEB on the number of packets. This fitting model implies that when the payload size  increases by $100$B, the average increase in total time for LLB is roughly $700 \mu s$ more than the increase experienced by BEB. The increase rate is statistically significant ($p$-value less than $0.001$). \vspace{-5pt}

%%%%%%%%%%%%%%%%%%%%%%%%%%%%%%%%%%%
%%%%%%%%%%%%%%%%%%%%%%%%%%%%%%%%%%%
%%%%%%%%%%%%%%%%%%%%%%%%%%%%%%%%%%%

\section{Discussion} \vspace{-2pt}

We have presented our evidence for why assumption A2 is flawed. In this section, we conclude our argument by considering a few unresolved observations, and discussing the legitimacy of our findings in the context of other protocols/networks. \vspace{-2pt}
 
\subsection{Oddities at Small Scale} \label{sec:large}
 
A few issues remain unaddressed:  \vspace{-0pt}
\begin{enumerate}[label=(\roman*)]
\item In terms of asymptotic bounds on CW slots, the newer algorithms are ordered ``best'' to ``worst'' as STB, LLB, and LB. Yet,  Figure~3 shows LB outperforming LLB. % and  competing with STB. \vspace{-0pt}

\item In terms of asymptotic bounds on the number of collisions, the newer algorithms are ordered ``best'' to ``worst'' as STB, LLB, and LB.  Yet, Figure~11 shows STB suffering a larger number of ACK timeouts than both LLB and LB. \vspace{-0pt}

\item BEB and STB have an asymptotically equal number of collisions, but STB is expected to suffer more and this is supported by Figures~7 and 8. What is the long-term behavior? \vspace{-0pt}
\end{enumerate}

As we discussed previously in Section~\ref{sec:experimental}, NS3 is valuable in revealing flawed assumptions via the extraordinary level of detail it provides; however, this also prevents experimentation with NS3 at larger scales. We attempt to shed light on (i) - (iii) by examining larger values of $n$ in order to see if our predictions are met, and we employ our simpler Java simulation for this task. 

To address (i), we look at $n\leq 10^5$ as plotted in Figure~15. Now we see that STB performs best in terms of CW slots, and that LLB is indeed outperforming LB. This supports prior theoretical results for CW slots given sufficiently large $n$.

In regard to (ii), we again take $n\leq 10^5$ and plot the ratio of collisions: LB vs STB and LLB vs STB.  Figure~16 demonstrates that the number of collisions for LB quickly exceeds STB. The tougher case is LLB which only begins to evidence a greater number of collisions at approximately $n=30,000$. Nevertheless, we observe a trend towards exceeding parity, as expected. Moreover, the sluggish trajectory is not surprising given our analysis in Section~\ref{sec:theory}.  

Finally, for (iii), we observe that the number of collisions for STB is larger than BEB by roughly a factor of 2 over this large range of $n$. Note that the plot of BEB/STB is (roughly) flat, as expected from our asymptotic analysis of collisions.

%%%%%%%%%%%%%%%%%%%%%%%%%%%%%%%%%%%
%%%%%%%%%%%%%%%%%%%%%%%%%%%%%%%%%%%
%%%%%%%%%%%%%%%%%%%%%%%%%%%%%%%%%%% 

\begin{figure}[t]\vspace{-48pt}
\captionsetup[subfigure]{labelformat=empty}
\centering
\hspace{-0.7cm}\begin{subfigure}{0.26\textwidth}
\includegraphics[width=1.15\textwidth]{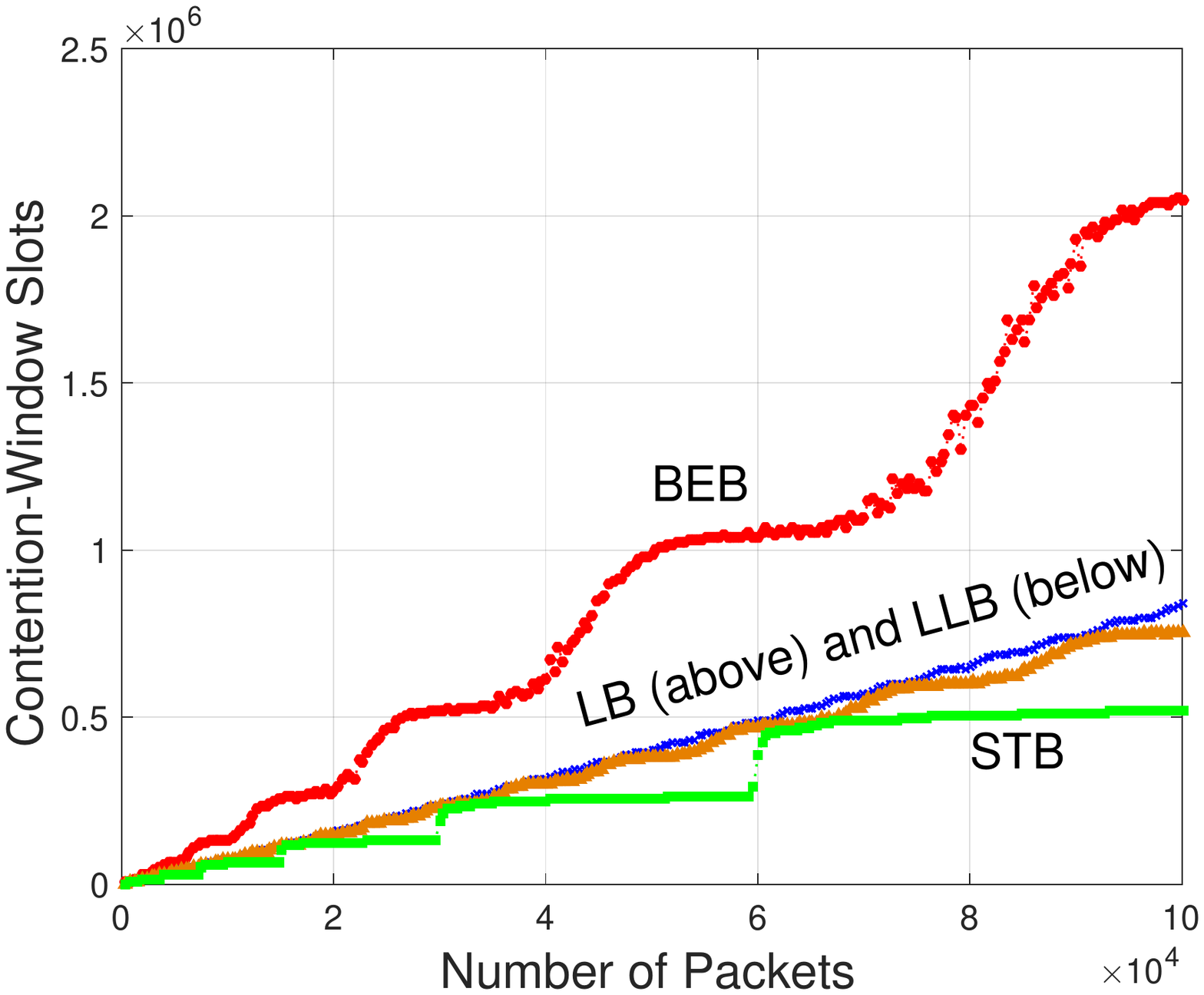}
\vspace{-2cm}\caption{\hspace{19pt}(15)}
\end{subfigure}\hspace{-0.2cm}
\begin{subfigure}{0.26\textwidth}
\includegraphics[width=1.15\textwidth]{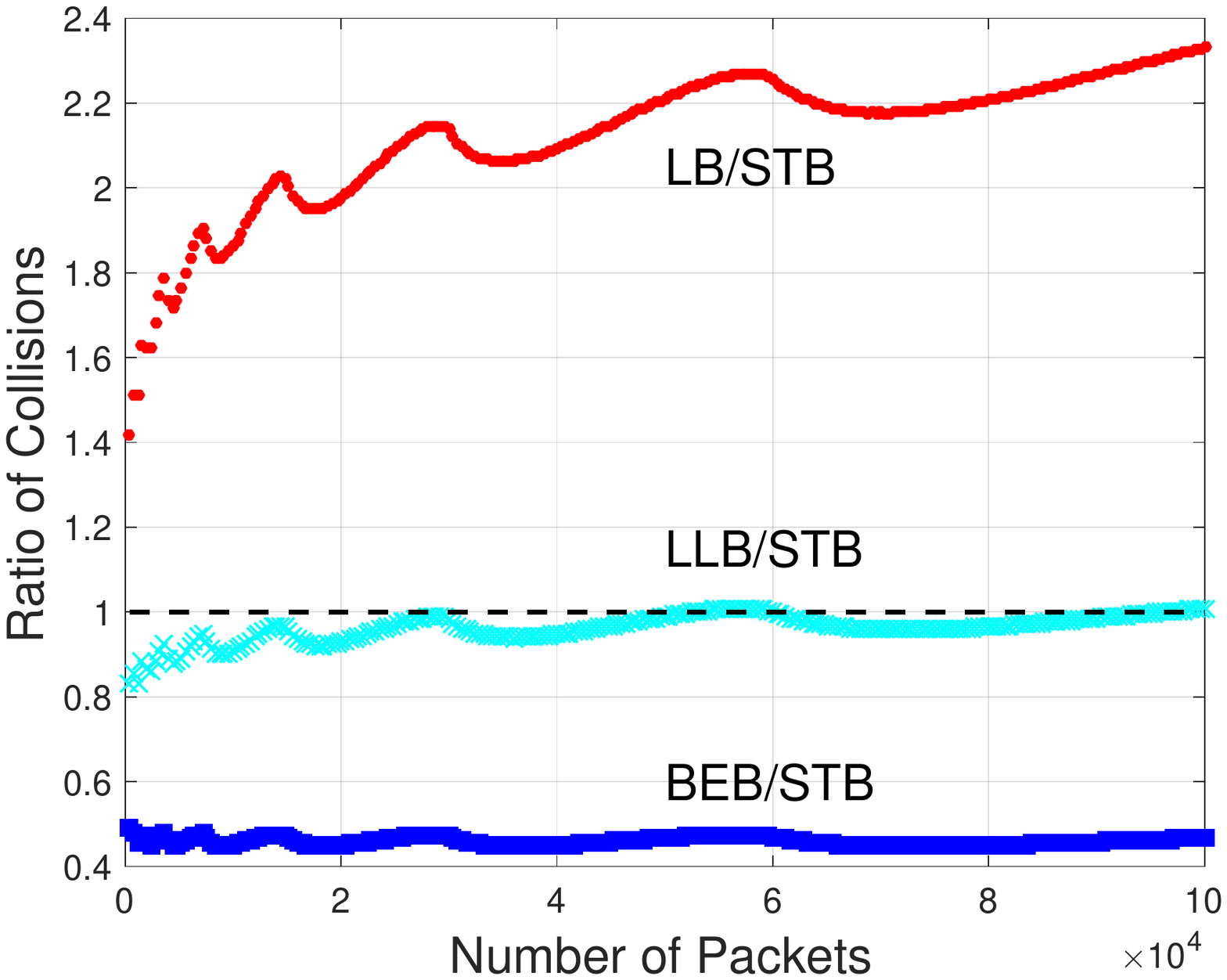} 
\vspace{-2cm}\caption{\hspace{19pt}(16)}
\end{subfigure}
\vspace{-0pt}\caption*{{\bf Figures 15-16.} Results from Java simulation results with $200$ trials per $n\leq 10^6$ in increments of $400$: (15) CW slots with median values plotted, (16) ratio of median number of collisions for BEB, LLB, and LB versus STB.}\label{fig:total-time}\vspace{-5pt}
\end{figure}

%%%%%%%%%%%%%%%%%%%%%%%%%%%%%%%%%%%
%%%%%%%%%%%%%%%%%%%%%%%%%%%%%%%%%%%
%%%%%%%%%%%%%%%%%%%%%%%%%%%%%%%%%%%
 
\subsection{Scope of Our Findings}\label{sec:interpret} 

In this section, we consider to what extent our findings are an artifact of IEEE 802.11g, and whether LB, LLB, STB might do better inside other protocols.\smallskip

\noindent{\bf IEEE 802.11g uses a truncated BEB, is this significant?} In our experiments, the maximum congestion-window size is $1024$ which differs from the abstract model where no such upper bound exists. However, even for $n=150$, this maximum is rarely reached during an execution of BEB and this does not seem to have any noticeable impact on the trend observed in Figures 3 and 4. 
\smallskip\smallskip

\noindent{\bf What if smaller packets are used?} During a collision, the time lost to transmitting would be reduced. In an extreme case, if the transmission of a packet fit within a slot, this would align more closely with A2.

Due to overhead, packet size has a lower bound in IEEE 802.11. Additionally, in NS3, there is a 12-byte payload minimum which translates into a minimum packet size of $76$ bytes for our experiments.\footnote{This is set within the \texttt{UdpClient} class of NS3.} The same qualitative behavior is observed in terms of CW slots and total time. For total time, the increase by LLB, LB, and STB is $6.6\%$, $17.8\%$, and $20.6\%$.

Alternatives to 802.11 might see more significant decreases. However, there is a tradeoff for {\it any} protocol. A smaller packet implies a reduced payload given the need for control information (for routing, error-detection, etc.) and this means that  throughput is degraded.\smallskip\smallskip

\noindent{\bf What if the ACK-timeout duration is reduced or acknowledgements are removed altogether?} This would also bring us closer to A2, although less so than having smaller packets -- the delay from ACK timeouts does not dominate as discussed in Section~\ref{sec:hidden}. In our experiments, the ACK-timeout is $75\mu$s (recall Section ~\ref{sec:experimental}) and values below this threshold will lead a station to consider its packet lost before the ACK can be received. This results in unnecessary retransmissions and, ultimately, poor throughput.

Totally removing acknowledgements (or some form of feedback) is difficult in many settings since, arguably, they are critical to {\it any} protocol that provides reliability; more so when transmissions are subject to disruption by other stations over a shared channel. \smallskip\smallskip

%%%%%%%%%%%%%%%%%%%%%%%%%%%%%%%%%%%
%%%%%%%%%%%%%%%%%%%%%%%%%%%%%%%%%%%
%%%%%%%%%%%%%%%%%%%%%%%%%%%%%%%%%%% 

\noindent{\bf To what extent do these findings generalize to other protocols?} We do {\it not} claim that our findings hold for all protocols. If (a) sufficiently small packets are feasible {\it and} (b) reliability is not paramount, performance should align better with theoretical guarantees derived from using assumption A2.

We {\it do} claim that the performance of how collision detection is performed -- and which is ignored under  A2 --  seems common to several other protocols. Examples include members of the IEEE 802.11 family,  IEEE 802.15.4 (for low-rate wireless networks), and IEEE 802.16 (WiMax). These employ some form of backoff and, regarding (a) and (b), each incurs header bloat and uses feedback via acknowledgements or a timeout to determine success or failure.  This is a significant slice of current wireless standards that, given our findings, could {\it potentially} experience performance degradation if BEB is replaced by LB, LLB, and STB.\vspace{-5pt}

\begin{figure}[h]
\begin{center}
{
\fbox{\colorbox{light-gray}{
\begin{minipage}[h]{0.42\textwidth} 
\noindent{\bf Result 6.} {\it Designing contention-resolution algorithms using assumption A2 seems likely to translate into poor performance in practice for a range of protocols.}
\end{minipage}
          }
     }     
}
\end{center}
\vspace{-10pt}
 \end{figure}

A setting where the abstract model may be valid is networks of multi-antenna devices. If a collision can be detected more efficiently, perhaps by a separate antenna,  the delay due to transmission time can be reduced. Canceling the signal at the sending device so that other transmissions (that would cause the collision) can be detected is challenging. However, this  is possible (for an interesting application, see~\cite{Gollakota:2011:THY:2018436.2018438}) and such schemes have been proposed using multiple-input multiple-output (MIMO) antenna technology~\cite{6963622,6962145}.

Finally, we note that future standards may satisfy (a) and (b). A possible setting is the Internet-of-Things (IoT);  for example,~\cite{7524360} characterizes  IoT transmissions as ``small'' and ``intermittent, delay-sensitive, and short-lived". To reduce delay, the authors argue for removing much of the control messaging used by traditional MAC protocols. Therefore, this setting seems more closely aligned with A2. However, using this same logic,~\cite{7524360} also argues for the removal of any backoff-like contention-resolution mechanism. Nevertheless, these standards  are  in flux and we may see protocols that avoid the issues we identify here.

%%%%%%%%%%%%%%%%%%%%%%%%%%%%%%%%%%%
%%%%%%%%%%%%%%%%%%%%%%%%%%%%%%%%%%%
%%%%%%%%%%%%%%%%%%%%%%%%%%%%%%%%%%%

\section{A Size-Estimation Approach}

Given our findings, we consider an alternative approach to the design of contention-resolution algorithms. Feedback is a useful ingredient  as it allows stations to tune their sending probabilities. For windowed algorithms, this feedback is obtained via collisions which, as we have seen, is  costly.

To avoid this problem, we examine a different approach. Stations first estimate $n$ and then execute \defn{fixed backoff} where the size of each contention window is set to this one-time estimate.  So long as the algorithm avoids an underestimate, the large number of collisions incurred by BEB, LB, LLB and STB should be avoided.  

Work in~\cite{GreenbergFlLa87,bender:contention,cali:dynamic,cali:design,bianchi:kalman} examines size estimation as a means for improving performance, although the methods and traffic assumptions differ (see Section~\ref{sec:related}). We aim to experiment with an algorithm for a single batch of arrivals, whose specification lends itself to implementation, and whose improved performance manifests for practical values of $n$.

To this end, the size-estimation component of our algorithm, \textsc{Best-of-}$k$, specified in Figure~17 is a variant of a well-known ``folklore'' result (see~\cite{jurdzinski:energy}). For $k=\Theta(1)$, a significant overestimate may occur, but the amount by which it can underestimate is bounded; w.h.p. the estimate will be $\Omega(n/\log n)$. %Increasing $k$ leads to better provable guarantees, but we are interested in performance.

%%%%%%%%%%%%%%%%%%%%%%%%%%%%%%%%%%%
%%%%%%%%%%%%%%%%%%%%%%%%%%%%%%%%%%%
%%%%%%%%%%%%%%%%%%%%%%%%%%%%%%%%%%% 

 \begin{figure}[t]
\begin{mdframed}
\begin{center}

\selectfont

\fbox{\hspace{-7pt}\colorbox{light-gray}{
\begin{minipage}[t]{0.99\textwidth} 

	\noindent{}\textsc{Best-of-$k$}\vspace{-3pt}

\begin{itemize}[leftmargin=3mm]

\item For $i = 0$ to $10$,  do:\vspace{-0pt}

\begin{itemize}[leftmargin=3mm]

\item For each of $k$ consecutive slots,  do:\vspace{-0pt}

\begin{itemize}[leftmargin=3mm]\renewcommand{\labelitemiii}{$-$}

\item With probability $\frac{1}{2^i}$, send dummy packet; otherwise, sense the channel.

\item If channel is clear for more than $k/2$ slots, then:
	
\begin{itemize}[leftmargin=3mm]\renewcommand{\labelitemiii}{$-$}

\item $W \leftarrow 2^i$

\item Terminate and start executing fixed backoff using a contention window of size $W$

\end{itemize}
\end{itemize}

\end{itemize}

\end{itemize}

 \end{minipage}
          }
     }     
\end{center}
\end{mdframed}
\vspace{-5pt}
\caption*{{\bf Figure 17.} The \textsc{Best-of-k} contention-resolution algorithm.}\label{fig:generic-backoff}\vspace{-3pt}
 \end{figure} 
 %%%%%%%%%%%%%%%%%%%%%%%%%%%%%%%%%%%
%%%%%%%%%%%%%%%%%%%%%%%%%%%%%%%%%%%
%%%%%%%%%%%%%%%%%%%%%%%%%%%%%%%%%%% 

Dummy packets of $28$ bytes are used in the size-estimation phase; this small size is possible because the packets contain none of the upper-layer headers (these are not used in our IEEE 802.11 experiments since routing requires the upper-layer headers).  Execution proceeds in $35\mu$s rounds during which a dummy packet is transmitted. Channel sensing is used to distinguish ``busy'' from ``clear''; therefore, we avoid any collision detection and the use of any acknowledgements for these dummy packets.

As expected, for $k=3$, the estimates are somewhat noisy, but this improves with $k=5$; see Figure~18. Notably, increasing $k$ does not significantly impact performance since the time required to run the size-estimation component is negligible (less than $5\%$) of the total time; instead, running fixed backoff is the main source of delay. We also observe that only overestimates occur, as predicted. This has the benefit of yielding good performance due to the lack of collisions. As demonstrated by Figure~19, both versions of the size-estimation approach outperform BEB, with $k=3$ and $k=5$ yielding a decrease in total time of  $26.0\%$ and $24.7\%$ respectively.\vspace{-5pt}

\begin{figure}[h]
\begin{center}
{
\fbox{\colorbox{light-gray}{
\begin{minipage}[h]{0.42\textwidth} 
\noindent{\bf Result 7.} {\it A size-estimation approach to designing contention-resolution algorithms appears promising.}
\end{minipage}
          }
     }     
}
\end{center}
\vspace{-10pt}
 \end{figure}

%%%%%%%%%%%%%%%%%%%%%%%%%%%%
%%%%%%%%%%%%%%%%%%%%%%%%%%%%
%%%%%%%%%%%%%%%%%%%%%%%%%%%%

\begin{figure}[t]
\captionsetup[subfigure]{labelformat=empty}
\centering
\vspace{-1.85cm}
\hspace{-0.7cm}\begin{subfigure}{0.26\textwidth}
\includegraphics[width=1.15\textwidth]{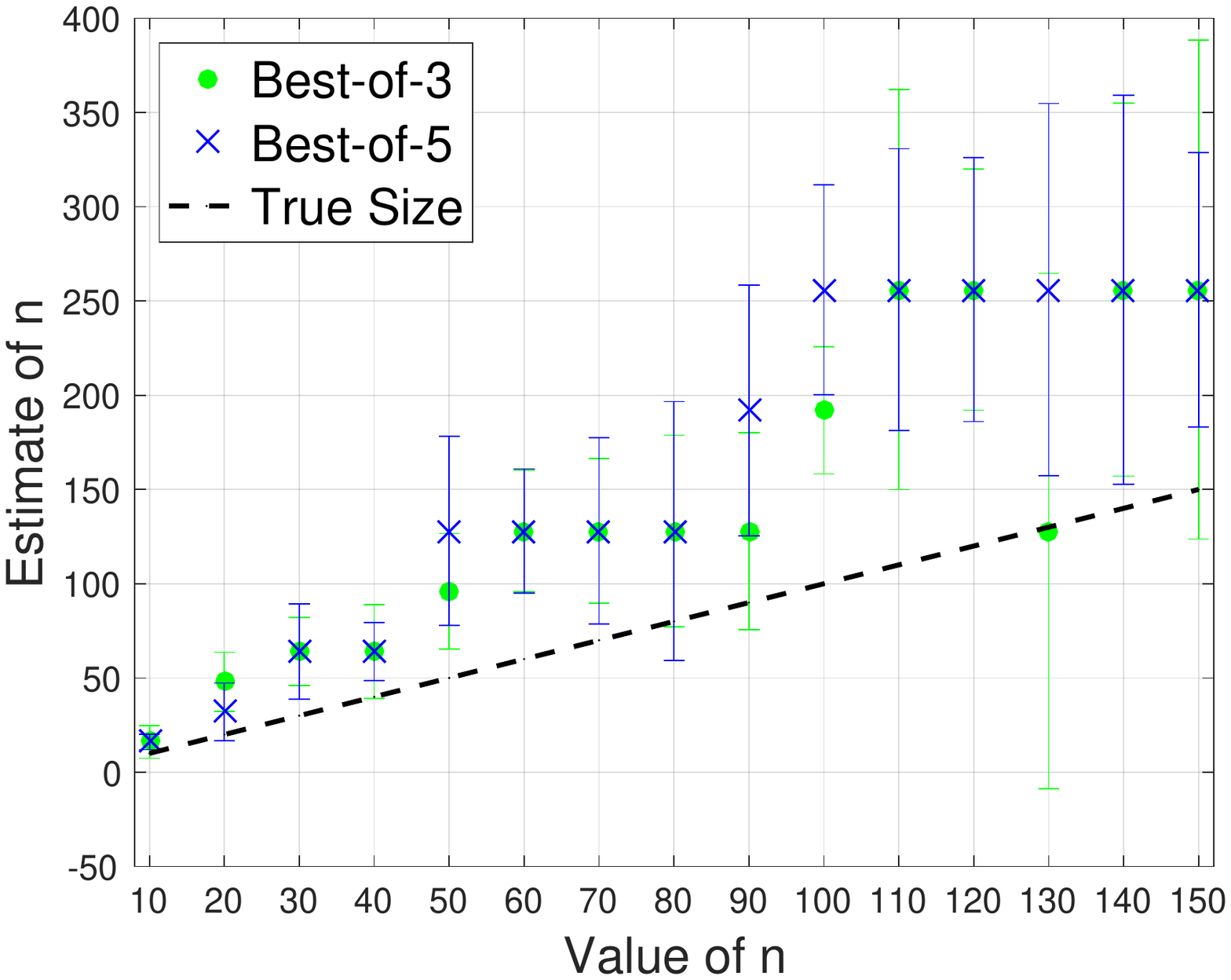}
\vspace{-2cm}\caption{\hspace{19pt}(18)}
\end{subfigure}\hspace{-0.2cm}
\begin{subfigure}{0.26\textwidth}
\includegraphics[width=1.15\textwidth]{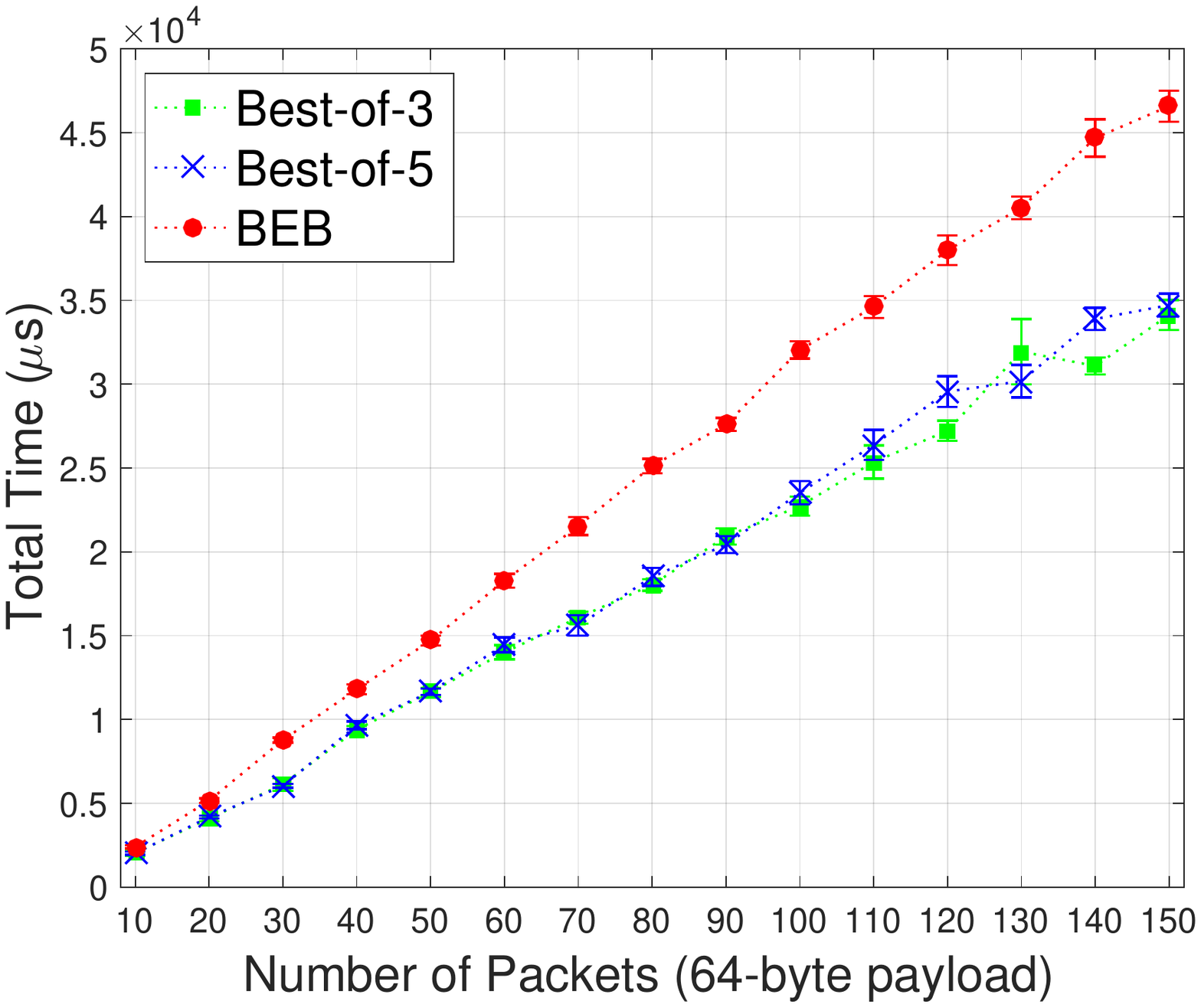} 
\vspace{-2cm}\caption{\hspace{19pt}(19)}
\end{subfigure}
\vspace{-0pt}\caption*{{\bf Figures 18-19.} Median values are reported for NS3 results with $20$ trials for each value of $n$: (18) estimation values, (19) total time ($\mu$s) for BEB, Best-of-3, and Best-of-5. Bars represent $95\%$ confidence intervals.}\vspace{-0.5cm}
\end{figure}

Finally,  we note our assumption of synchronization might not hold in a dynamic setting. Furthermore, a real-world deployment would likely be ``messier'' with respect to interference from other devices/networks running different applications, impact of terrain and weather on transmissions,  etc. and our simple setup does not account for such phenomena. Additional methods would be needed to address these issues. However, as an initial proof of concept, our results suggest an approach to contention resolution that {\it may} compete with BEB in the single-batch case.

%%%%%%%%%%%%%%%%%%%%%%%%%%%%%%%%%
%%%%%%%%%%%%%%%%%%%%%%%%%%%%%%%%%
%%%%%%%%%%%%%%%%%%%%%%%%%%%%%%%%%
 
\section{Related Work}\label{sec:related}

Exponential backoff has been studied under Poisson-distributed traffic~(see~\cite{GoldbergMa96a,GoldbergMaPaSr00,HastadLeRo87,RaghavanUp95}). Guarantees on stability are known~\cite{Al-Ammal2000,Al-Ammal2001,Goodman:1988:SBE:44483.44488}, and under saturated conditions~\cite{bianchi:performance}.

There is a vast body of literature addressing the performance of IEEE 802.11 (for examples, see~\cite{Ni:survey,Kuptsov201437,5191039,1543687,LiHSC03,duda:understanding,5963276}). There are several results that focus on the performance of BEB within IEEE 802.11; however, they do not address issues of the abstract model, bursty traffic, or the newer algorithms examined here. Nonetheless, we summarize those works that are most closely related.

Under continuous traffic, windowed backoff schemes are examined in~\cite{sun:backoff} with a focus on the tradeoff between throughput and fairness. The authors focus on polynomial backoff and demonstrate via analysis and NS2 (the predecessor to NS3) simulations that quadratic backoff is a good candidate with respect to both metrics.

Work in~\cite{1424043} addresses saturated throughput (each station always has a packet ready to be transmitted) of exponential backoff; roughly, this is the maximum throughput under stable packet arrival rates. Custom  simulations are used to confirm these findings. 

In~\cite{6859627}, the authors propose backoff algorithms where the size of the contention window is modified by a small constant factor based on the number of successful transmissions observed. NS2  simulations are used to demonstrate improvements over BEB within 802.11 for a steady stream of packets (i.e. non-bursty traffic). 

Lastly, in~\cite{saher:log}, the authors examine a variation on backoff where the contention window increases multiplicatively by the logarithm of the current window size (confusingly, also referred to as ``logarithmic backoff''). NS2  simulations imply an advantage to their variant over BEB within IEEE 802.11, again for non-bursty traffic. 

In regard to size-estimation approaches, there is prior work on tuning the probability of a transmission under Poisson-distributed traffic~\cite{hajek:decentralized,kelly:decentralized,Gerla:1977}. Subsequent work in~\cite{cali:dynamic,cali:design} offers (custom) simulation and analytical results on performance improvements assuming that the transmission interval for which a station backs off is sampled from the geometric distribution. In~\cite{bianchi:kalman}, the authors propose a method for estimating the number of contending stations under saturation conditions, and custom simulations demonstrate the accuracy of this approach. More recent work in~\cite{bender:contention} proposes a size-estimation scheme with small (asymptotic) sending and listening costs; however, no experimental results are provided and implementing this scheme may be challenging.

Regarding the time required for a single successful transmission,~\cite{willard:loglog} demonstrates a lower-bound of $\Omega(\log\log n)$. In different communication models, other bounds are known~\cite{fineman:contention2,fineman:contention}. 

A class of tree-based algorithms for contention resolution is proposed in~\cite{Capetanakis:2006}. Work in~\cite{bender:heterogeneous} addresses the case of heterogeneous packet sizes. The case where packets can arrive dynamically is examined  in~\cite{chlebus:better,chlebus:wakeup,chrobak:wakeup}.

Energy efficiency is important to multiple access in many low-power wireless networks~\cite{jurdzinski:energy,chang:exponential,bender:contention}. When the communication channel is subject to adversarial disruption, several results address the challenge of multiple access ~\cite{awerbuch:jamming,richa:jamming2,richa:jamming3,richa:jamming4,
ogierman:competitive,richa:efficient-j,richa:competitive-j,bender:how,Tan2014}.  Finally, deterministic broadcast protocols have also received significant attention~\cite{ChlebusKoRo06,ChlebusKoRo12,anantharamu:adversarial-opodis}. 

\section{Concluding Remarks}

We have presented evidence that a model commonly used for designing contention-resolution algorithms is not adequately accounting for the cost of collisions. A number of interesting questions remain. 

In terms of analytical work, we have argued for why collisions matter at small scale and asymptotically, but what is the optimal tradeoff between collisions and CW slots? Does this change when we consider multi-hop networks or long-lived bursty traffic? Assuming that this tradeoff is known, can we design algorithms that leverage this information? 

Regarding future experimental work, it may be of interest to perform a similar evaluation on other protocols. For example, much of what is examined in this work seems to apply to contention resolution under IEEE 802.15.4, and we expect collisions to be similarly expensive. However, are there subtle differences in the protocol that allow IEEE 802.15.4 to avoid collisions? What about newer wireless standards? Understanding any such behavior may aid in the design of future contention-resolution algorithms.
\smallskip\smallskip

\noindent{\bf Acknowledgements.} We are grateful to David Dampier for providing us with access to the computing resources at HPC$^2$.\vspace{-0pt}

%\bibliographystyle{IEEEtran}
%\bibliography{jam,backoff,more-backoff}

% Generated by IEEEtran.bst, version: 1.14 (2015/08/26)

\end{document}